\documentclass{article}
\usepackage{amssymb}
\usepackage[numbers]{natbib}
\usepackage{enumitem}
\usepackage{pdfpages}
\usepackage{amssymb}
\usepackage{subeqnarray}
\usepackage{float}
\usepackage{pifont}
\usepackage{graphicx}
\usepackage{epstopdf}
\usepackage{url}
\usepackage{amsmath}
\usepackage{amsthm}
\usepackage[utf8]{inputenc}
\usepackage[T1]{fontenc}
\usepackage{lmodern}
\usepackage{ae}
\usepackage[normalem]{ulem}

\newtheorem{theorem}{\bf Theorem}[section]
 \newtheorem{proposit}[theorem]{\bf Proposition}
 \newtheorem{coro}[theorem]{\bf Corollary}
\newtheorem{lem}[theorem]{\bf Lemma}

 \newtheorem{example}[theorem]{\bf Example}

\def\thm#1\par{\medskip\par\noindent\begin{theorem} \strut \sl #1 \end{theorem}\par}
\def\propo#1\par{\medskip\par\noindent\begin{proposit} \strut \sl #1 \end{proposit}
\par}
\def\cor#1\par{\medskip\par\noindent\begin{coro} \strut \sl #1 \end{coro}\par}
\def\lm#1\par{\medskip\par\noindent\begin{lem} \strut \sl #1 \end{lem}\par}
\date{}

\title{Loopless Algorithms to Generate Maximum Length Gray Cycles wrt. $k$-Character Substitution}

\author{Jean N\'eraud\\
{\small{\it Univ Rouen Normandie, LITIS UR 4108, F-76000 Rouen, France}}}
\begin{document} 

\maketitle

\begin{abstract}
Given a binary word relation $\tau$ onto $A^*$ and a finite language $X\subseteq A^*$, a $\tau$-Gray cycle over $X$ consists in a permutation $\left(w_{[i]}\right)_{0\le i\le |X|-1}$ of $X$
such that each word $w_{[i]}$ is an image  under $\tau$ of the previous word $w_{{[i-1]}}$.
We define  the complexity measure $\lambda_{A,\tau}(n)$, equal to the largest  cardinality of a language $X$ having words of length at most $n$, and st. some $\tau$-Gray cycle  over $X$  exists. 
The present paper is concerned with  $\tau=\sigma_k$,
the so-called $k$-character substitution,
st. $(u,v)\in\sigma_k$ holds  if, and only if, the Hamming distance of $u$ and $v$ is $k$.
We present loopless (resp., constant amortized time) algorithms for computing specific maximum length $\sigma_k$-Gray cycles. 
\end{abstract}

\section{Introduction}
\label{Intro}

In the framework of combinatorial algorithms, one of the most well-documented issues concerns the development of methods for generating, once and for all, each object of a specific class.  \cite{LEH64}. 
Many topics are concerned by such a problem: suffice it to mention sequence counting \cite{BLPP99},
signal encoding \cite{L81}, and  data compression \cite{R86}.

\smallbreak
In the whole paper we fix some alphabet, say $A$, and we assume that $|A|$, the cardinality of $A$, is not less than 2.
The so-called {\it binary Gray codes} 
first appeared in \cite{G58}:
given a binary alphabet $A$ and some positive integer $n$,
such objects referred to  sequences  with maximum length of pairwise different $n$-tuples  of characters (that is, words in  $A^n$), provided that any pair of consecutive items  differ by exactly one character.
Shortly after, a similar study was drawn in the framework of non-binary alphabets \cite{E84}.
Regarding other famous combinatorial classes of objects, the  term of {\it combinatorial Gray code}, for its part,
appeared in \cite{JWW80}:
actually,  the difference between successive items, although being fixed, need not to be small \cite{S97}. 
Generating all permutations of  a given $n$-element set
constitutes a noticeable example \cite{BV02,E73,V03,W09}.
The so-called bubble languages \cite{BFLRS20,RSW12} are also involved,  as well as  cross-bifix-free sets \cite{BBPSV14}, 
 Debruijn sequences \cite{FM78}, Dyck words \cite{VW06}, Fibonacci words \cite{BKV22},  Lyndon words \cite{V08},
Motzkin words \cite{V02}, necklaces \cite{RSW92,V08}, set partitions \cite{K76},  subsets of fixed size  \cite{EMcK84, JW80}:
the list is  far from exhaustive. For some surveys we suggest the reader report to \cite{T23,S97,vDH13,XM09}. 
From an algorithmic point of view, the ultimate feature  is to develop methods for producing each new object  with constant, or at least constant amortized time delay, that is 
 {\it loopless} or {\it constant amortized time} algorithms  are desired \cite{E73,LMNW23,SW13,V03,W09}.

\smallbreak
The Combinatorial Gray sequences are often required to be {\it cyclic} \cite{CDG92},
 in the sense that 
the initial term itself can be retrieved as  successor of the last one. 
Such a condition justifies the terminology of {\it Gray cycle} \cite[Sect. 7.2.1.1]{K05}.
In order to develop a formal framework, we note that each of the sequences we have  mentioned above is concerned with a binary word relation $\tau\subseteq A^*\times A^*$
($A^*$ stands for the free monoid generated by $A$).
For its part, the combinatorial class of objects can be modeled by some finite language $X\subseteq A^*$. 
Given  a sequence of words we denote in square brackets the corresponding indices: 
this will allow us to clearly distinguish the difference with $w_i$, the character in {\it position} $i$ in a given word $w$.
In addition, we set $\tau(w)=\{w':(w,w')\in\tau\}$.
We define a {\it Gray cycle over $X$ wrt. $\tau$} (for short: {\it $\tau$-Gray cycle over $X$}) as 
every finite sequence of words $\left(w_{[i]}\right)_{i\in [0,|X|-1]}$ satisfying each of the three following conditions:
\begin{enumerate}[label={\rm (G\arabic*)}, leftmargin=2cm]
\item 
For every word $x\in X$, some $i\in [0,|X|-1]$ exists st. we have $x=w_{[i]}$;
\label{iva}
\item
\label{ivb}
 For every $i\in [1,|X|-1] $,  we have $w_{[i]}\in\tau\left(w_{[i-1]}\right)$; in addition, the cond.  $w_{[0]}\in\tau\left(w_{[|X|-1]}\right)$ holds;
\item 
\label{ivc} For every pair $i,j\in [0,|X|-1]$, $i\neq j$ implies $w_{[i]}\neq w_{[j]}$.
\end{enumerate}
With this definition, 
in the Gray cycle the terms may have a variable length.
For instance, given the alphabet $A=\{0,1\}$, take for $\tau$  the binary word relation $\Lambda_1$ which,
with every word $w$, associates all the strings located  within a {\it Levenshtein distance} of $1$ from $w$ (see e.g.  \cite{N21}).
Actually  the  sequence $(0,00,01,11,10,1)$ is a $\Lambda_1$-Gray cycle over $X=A\cup A^2$.  

\smallbreak
In addition to the topics we mentioned above, two other  fields  involved by those Gray cycles should be mentioned.
Firstly,  regarding graph theory, a $\tau$-Gray cycle over $X$ exists iff. 
there is some Hamiltonian circuit  in the  graph of the relation $\tau$ (see \cite{S97} and for some surveys on such a notion \cite{G91,KM09,SW18,LPR19}). Secondly, 
the existence of a  $\tau$-Gray cycle over $X$ implies $\tau(X)\subseteq X$ (with $\tau(X)=\{\tau(w):w\in X\}$): as defined in   \cite{N21}, $X$ is $\tau$-{\it closed}.
Actually, closed sets  constitute a special subfamily in the famous {\it dependence systems} \cite{JSY94}.
However,   the fact that $X$ is $\tau$-{\it closed}
does not guarantees that  some  $\tau$-Gray cycle may exist  over $X$. 
A typical example is provided by $\tau=id_{A^*}$, the identity over $A^*$.
Indeed,  although every finite set $X\subseteq A^*$ is $\tau$-closed,
non-empty $\tau$-Gray cycle can only exist over $X$  if  $|X|=1$.

\smallbreak
In the present paper, we consider
the family of  all sequences 
that can be a $\tau$-Gray cycle over 
some subset $X$ of  $A^{\le n}$ (with $A^{\le n}=\{w\in A^*: |w|\le n\}$).
This is a natural question to focus on those sequences of maximum  length:
clearly, they correspond to subsets $X$ of maximum cardinality.  
We denote by $\lambda_{A,\tau}(n)$ that maximum length. 
This actually means introducing some  complexity measure for the binary word relation $\tau$ \cite{Ne21}.
We focus on  the case where $\tau$ is  $\sigma_k$, the so-called {\it $k$-character substitution}.
With every word with length at least $k$, say $w$, this relation associates all  the words $w'$, 
with $|w'|=|w|$, and st. the character $w'_i$ differs from $w_i$ in exactly $k$ values of  $i\in[1,|w|]$:
in other words, the Hamming distance of $w$ and $w'$ is $k$. 

\smallbreak
Some words on the word binary relation $\sigma_k$: firstly, as commented in \cite{JK97,N21}, this relation  has noticeable inference in the famous framework of {\it error detection}.
Secondly, the cond. $w'\in\sigma_k(w)$ implies $|w'|=|w|$, therefore 
if  there is some $\sigma_k$-Gray cycle over a language $X$, then all the words in $X$ have a common length: by definition $X$ is a {\it uniform} set. 
From this point of view, the classical Gray codes that allow to generate all $n$-tuples over $A$,  
correspond to $\sigma_1$-Gray cycles over $A^n$. 
Actually, in the case where $A$ is  a binary alphabet, some maximum length $\sigma_2$-Gray cycles  have also been constructed
(we have $\lambda_{A,\sigma_2}(n)=2^{n-1}$) \cite[Exercice 8, p. 77]{K05}. 
However, in the most general case, although an exhaustive description of $\sigma_k$-closed variable-length codes has been  provided in \cite{N21},
the question of generating some $\sigma_k$-Gray cycle of maximum length  has remained open.
The present paper present algorithmic constructions to generate those sequences of maximum length.
More precisely, we establish the following result:
\medbreak\noindent
 {\bf Theorem}
 {\it  
Given a finite alphabet  $A$, $k\ge 1$, and $n\ge k$, 
there is a loopless algorithm that allows to generate some maximum length $\sigma_k$-Gray cycle.
In addition  the following equation  holds:}
\begin{eqnarray}
\lambda_{A,\sigma_k}(n)=\left\{
\begin{array}{ccccc}
|A^n|&|A|\ge 3, n\ge k\\
2&|A|=2, n=k\\
|A|^n&~|A|=2, n\ge k+1, k~ {\rm is~ odd}\\
~~~ |A|^{n-1}&~~~ |A|=2, n\ge k+1,  k~ {\rm is~ even.}
\end{array}
\right.
\nonumber
\end{eqnarray}

\smallbreak
Beforehand, the computation of such maximum length $\sigma_k$-Gray cycles  is done thanks to induction-based equations.
In addition, in each case  we present an iteration-based method for computing those sequences: it directly allows to compute the term of index $i$ by starting from the term of index $i-1$.

\bigbreak
We now shortly describe the contents of the paper:

\smallbreak
-- Section \ref{Prelim}, 
is devoted to the preliminaries. We fix some complementary definitions and notations; in addition, we recall the  two famous examples of
the  {\it binary} (resp., $|A|$-{\it ary}) {\it reflected Gray code}. 

\smallbreak
-- In Sect. 2, 
 we focus on the case where the alphabet $A$ possesses  at least $3$ letters.
Starting with the $|A|$-{\it ary} reflected Gray code, by  establishing some induction formula we prove that, given a pair of  positive integers $n,k$,   there is a  peculiar Gray cycle, namely $h^{n,k}=\left(h^{n,k}_{[i]}\right)_{0\le i\le |A|^n-1}$: its length is $|A|^n$.
\smallbreak
-- An iteration-based method for computing the preceding cycle is developed in Sect. 3. 

\smallbreak 
-- In Sect. 4, 
in the case where $A$ is a binary alphabet, with $k$ being an odd integer, we also
compute a maximum length $\sigma_k$-Gray cycle. Once more this is done by establishing some
inductive method: practically it relies on two peculiar $k$-Gray cycles.

\smallbreak
-- A corresponding iteration-based method of computation is developed in Sect. 5. 

\smallbreak
-- At last, in Sect. 6, 
in the case where $A$ is a binary alphabet, with $k$ being an even positive integer,
we  also compute a  maximum length Gray cycles : this leads to complete the proof of the  theorem we mentioned above.

In addition, it is common in the literature to define a $k$-Gray code as the sequence
where two consecutive items have distance at most $k$ \cite{T23}. In the case of the Hamming distance, this notion corresponds to the so-called $\Sigma_k$-Gray cycles, where the relation $\Sigma_k$ st. $(w,w')\in\Sigma_k$ iff. the Hamming distance of $w$ and $w'$ is not greater than $k$.
We discuss where our results intersect with such a topic.
Some further development is  also raised.
\section{Preliminaries}
\label{Prelim}
Several definitions and notation  have already been fixed.
In  the whole paper, $A$ stands for a finite alphabet, with $|A|\ge 2$.
Given a word $w\in A^*$,  we denote by $|w|$ its length;  in addition, for every $a\in A$, we denote by $|w|_a$  the number of occurrences of the character $a$ in $w$.
Given a pair of words $w,w'\in A^*$, $w'$ is a {\it prefix} (resp., {\it suffix}) of $w$ if some word $u\in A^*$ exists st. $w=w'u$ (resp. $w=uw'$).
Given a word $w\in A^*$, $m\in [0,|w|]$, we denote by $A^{-m}w$ the unique suffix of $w$ with length $|w|-m$. 

Historically Gray cycles have been developed in order to generate integers in $|A|$-ary numeration system.  From this point of view, regarding the position of the characters in the word $w$, 
it is convenient to set $w=w_n\cdots w_1$, with $w_i\in A$, for every $i\in [1,n]$ (we say that $w_i$ is the character with {\it position} $i$ in the word $w$).
In other words, in the $|A|$-ary numeration system 
 the word $w=w_n\cdots w_1$ is the representation of the integer  $w_n|A|^{n-1}+w_{n-1}|A|^{n-2}+\cdots +w_1$.  

\bigbreak\noindent
{\it The reflected binary Gray cycle}\\
Let $A=\{0,1\}$ and $n\ge 1$. The most famous example of $\sigma_1$-Gray cycle over $A^n$ is certainly the  so-called {\it reflected binary Gray code} (see e.g. \cite[pp. 5-6]{K05} or \cite[Sect. 3.1]{T23}):
 in the present paper we denote it by $g^{n,1}$. It can be computed in different ways:
 
\smallbreak
-- Firstly, the sequence can be defined recursively by the following rule:
\begin{eqnarray}
\label{BGRC-recursion}
g^{0,1}=(\varepsilon);~~~g^{n+1,1}=\left(\left(0g^{n,1}\right),\left(1\left(g^{n,1}\right)^R\right)\right)
\end{eqnarray}
In this notation  the comma stands for the sequence concatenation.
Moreover, given a finite sequence, say $x=\left(x_1,\cdots x_n\right)$ and a word $w\in A^*$, we set $x^R=\left(x_n,\cdots,x_1\right)$ and $wx=\left(wx_1,\cdots, wx_n\right)$.
Actually  Eq. (\ref{BGRC-recursion}) leads to construct the whole sequence $g^{n+1,1}$ by applying a series of one character concatenations on the left over the words of $g^{n,1}$.

\smallbreak
-- Secondly, in the literature there is a  famous   constant amortized-time  iterative algorithm (in the paper we denote it by Algorithm (a)) that allows to compute $g^{n,1}$. 
The method starts by  setting $g^{n,1}_{[0]}=0^n$. After that, for every $i\in [1,|A|^n-1]$, the word  $g^{n,1}_{[i]}$ is computed from right to left by starting from $g^{n,1}_{[i-1]}$. Actually 
  a unique integer $j\in [1,n]$ exists st.  in both  words $g^{n,1}_{[i]}$ and  $g^{n,1}_{[i-1]}$, the corresponding characters in position $j$ differ:
with the preceding convention over character positions,  $j$  is chosen in order to satisfy the following condition:
\begin{eqnarray}
\label{choix-j-g}
j~{\rm is~ the~ {\it minimum}~position~in~}g^{n,1}_{[i]}{\rm ~  st.~}g^{n,1}_{[i]}\notin\{g^{n,1}_{[0]},\cdots,g^{n,1}_{[i-1]}\}.
\end{eqnarray}
\begin{example}
\label{Ex1}
Below are column representations of the sequences $g^{2,1}$ and $g^{3,1}$:
$$
\begin{array}{c}
 g^{2,1}\\ \overbrace{}\\00\\ 01\\11\\~10
\ \\ \ \\ \ \\ \ \\ \ \\
\end{array}
~~~~~~~~~~~~~~~~
\begin{array}{c}
 g^{3,1}
\\ \overbrace{~~~~~~}\\000\\ 001\\ 011\\ 010\\ 110\\ 111 \\101\\ 100\\ 
\end{array}
$$
\end{example}
By construction, for every $n\ge 1$, each of the following identities holds:
\begin{eqnarray}
\label{g-termes-remarquables}
g^{n,1}_{[0]}=0^n,~~g^{n,1}_{[1]}=0^{n-1}1,~~g^{n,1}_{[2^{n}-2]}=10^{n-2}1,~~g^{n,1}_{[2^{n}-1]}=10^{n-1}.
\end{eqnarray}

\bigskip\noindent
{\it The  $|A|$-ary reflected Gray cycle}\\
The preceding constructions can be extended in order to obtain the so-called {\it $|A|$-ary reflected Gray code} over $A^n$, that we denote by $h^{n,1}$.  
More precisely we set $A=\{0, \cdots,p-1\}$ and we denote  by $\theta$  the cyclic permutation defined by $0\rightarrow  1,\dots, p-2\rightarrow p-1, p-1\rightarrow 0$. 

\smallbreak
 -- Firstly, $h^{n,1}$ can be defined recursively by the following rule \cite[Sect. 3.19]{T23}:
\begin{eqnarray}
\label{BGRC-recursion-gen}
h^{0,1}=(\varepsilon);~~~h^{n+1,1}=\left(0h^{n,1},1\left(h^{n,1}\right)^R, 2h^{n,1}, 3\left(h^{n,1}\right)^R,\cdots, (p-1)\Gamma\right)\\
{\rm with~~}
\label{CCg}
\Gamma=\left\{
\begin{array}{l}
h^{n,1}~~{\rm if}~~p ~~{\rm is~~ odd }\\
\left(h^{n,1}\right)^R ~~{\rm otherwise.}
\end{array}
\right.
\nonumber
\end{eqnarray}

-- Secondly, the sequence $h^{n,1}$ can be generated by applying a constant amortized-time algorithm \cite{C63,E84}, that we denote by Algorithm (b).
The method starts by setting  $h^{n,1}_{[0]}=0^n$.
For every $i\in [1,|A|^n-1]$, starting from $h^{n,1}_{[i-1]}$ the word $h^{n,1}_{[i]}$ is computed from right to left. 
More precisely, there are $j\in [1,p^{n_0}-1]$, $c,d\in A$, where  $c$ and  $d$ are the characters respectively in position $j$ in $h^{n,1}_{[i-1]}$ and  $h^{n,1}_{[i]}$,
st. both the  each of following conditions holds:
\begin{eqnarray}
\label{CCh}
 j {\rm~  is~ the~ {\it minimum}~ position ~in~} h^{n,1}_{[i]} {\rm~ st.}~ h^{n,1}_{[i]}\notin\{h^{n,1}_{[0]},\cdots,h^{n,1}_{[i-1]}\},\\
d=\left\{
\begin{array}{l}
\theta(c)~{\rm if}~ i\div p^j~
{\rm is~~ even}\\
\theta^{-1}(c) ~{\rm otherwise.}
\end{array}
\right.
\end{eqnarray}

\begin{example}
\label{Ex2}
For $A=\{0,1,2\}$ the sequence $h^{3,1}$ 
is the concatenation in this order of the three following subsequences:
$$
\begin{array}{c}
h^{3,1}
\\ \overbrace{~~~~~~}\\000\\ 001\\ 002\\ 012\\ 011\\ 010 \\020\\ 021\\  022
\end{array}
~~~~~~~~
\begin{array}{c}
\ \\ \ \\ 122\\ 121\\ 120\\ 110\\ 111\\ 112 \\102\\  101\\100
\end{array}
~~~~~~~~
\begin{array}{c}
\ \\ \ \\ 
200\\ 201\\ 202\\ 212\\ 211\\ 210 \\220\\ 221\\  222
\end{array}
$$
\end{example}
%
\section{The case where we have $|A|\ge 3$}
\label{A>3}
Let  $n\ge k\ge 1$, $p\ge 3$, and $A=\{0,1,\cdots,p-1\}$. In what follows, we indicate the construction of  a peculiar  $\sigma_k$-Gray cycle over $A^n$, namely  $ h^{n,k}$.  
This is done by applying some induction over $k\ge 1$: in view of that, we set $n_0=n-k+1$.

\smallbreak
-- The starting point corresponds to $h^{n_0,1}$, the  $p$-ary reflected Gray code over  $A^{n_0}$ as reminded  in Section \ref{Prelim}.

\smallbreak
 -- For the induction stage, by starting with  a $\sigma_{k-1}$-Gray cycle over $A^{n-1}$,  namely $ h^{n-1,k-1}$,
we compute the corresponding sequence $ h^{n,k}$ as indicated in the following:

\noindent
Let $i\in [0,p^{n}-1]$, and let $q\in [0,p-1]$, $r\in [0,p^{n-1}-1]$ be  the unique pair of non-negative integers st. $i=qp^{n-1}+r$. We  set:
\begin{eqnarray}
\label{Gamma-A-ge3}
h^{n,k}_{[i]}=h^{n,k}_{[qp^{n-1}+r]}=
\theta^{q+r}(0) h^{n-1,k-1}_{[r]}.
\end{eqnarray}
As shown in Example \ref{E3}, by construction the resulting sequence  $ h^{n,k}$ is the concatenation in this order of $p$ subsequences namely $C_0, \dots,  C_{p-1}$,
with $C_q=\left( h^{n,k}_{[qp^{n-1}+r]}\right)_{0\le r\le p^{n-1}-1}$, for each $q\in [0,p-1]$.
Since $\theta$ is one-to-one, given a pair of different integers $q,q'\in [0,p-1]$,
for every  $r\in [0,p^{n-1}-1]$,  in each of  the  subsequences $C_q$, $C_{q'}$, the words  $h^{n,k}_{[qp^{n-1}+r]}$ and $h^{n,k}_{[q'p^{n-1}+r]}$ only differ in their initial characters,
which  respectively are $\theta^{q+r}(0)$ and $\theta^{q'+r}(0)$.
In addition, since $h^{n-1,k-1}$ is a $\sigma_{k-1}$-Gray cycle over $A^{n-1}$, we have $\left | h^{n,k}\right |=p\left | h^{n-1,k-1}\right |=p^n$. 
\begin{example}
\label{E3} 
{\rm Let $A=\{0,1,2\}$,  $n=3$, $k=2$, thus  $p=3$, $n_0=2$. By starting with $h^{n-1,k-1}=h^{2,1}$,
we construct the sequence  $ h^{n,k}$ as the concatenation of $C_0$, $C_1$, and $C_2$: 
$$
\begin{array} {c}
 h^{n-1,k-1}\\ \overbrace{~~~~~}\\00\\ 01\\ 02\\ 12\\ 11\\ 10 \\20\\ 21\\~22
\ \\ \ \\
\end{array}
~~~~~~~~~~~~~~~~
\begin{array}{c}
 h^{n,k}
\\ \overbrace{~~~~~~}\\000\\ 101\\ 202\\ 012\\ 111\\ 210 \\020\\ 121\\ 222\\
\\
\end{array}
~~~~~
\begin{array}{c}
\ \\
100\\ 201\\ 002\\ 112\\ 211\\ 010 \\120\\ 221\\ 022\\
\end{array}
~~~~~
\begin{array}{c}
\ \\ 
200\\ 001\\ 102\\ 212\\ 011\\ 110 \\220\\ 021\\ 122\\
\end{array}
$$
}
\end{example}
\begin{proposit}
\label{Cycle-Age3-n-k}
The sequence $h^{n,k}$
is a $\sigma_k$-Gray cycle over $A^n$.
\end{proposit}
\begin{proof}
We argue by induction over $k\ge 1$.
Regarding the base case, as indicated above  $h^{n_0,1}$
 is the $|A|$-ary reflected Gray sequence.
In view of the induction stage, we assume that the finite sequence $ h^{n-1,k-1}$ is a  $\sigma_{k-1}$-Gray cycle over $A^{n-1}$, for some $k\ge 2$.

\smallbreak\noindent
(i) We start by proving that $h^{n,k}$ satisfies Condition  \ref{ivb}. This will be  done through the three following steps:

\smallbreak
~~~~(i.i) Firstly, we prove that, for each $q\in [0,p-1]$, in  the subsequence $C_q$ 
two consecutive terms  are necessarily in correspondence under $\sigma_k$.
Given $r\in [0, p^{n-1}-1]$, by definition, we have $\theta^{r+q}(0)\in\sigma_1\left(\theta^{r+q-1}(0)\right)$.
Since $ h^{n-1,k-1}$ satisfies Condition \ref{ivb},  we have  $ h^{n-1,k-1}_{[r]}\in\sigma_{k-1}\left( h^{n-1,k-1}_{[r-1]}\right)$.
We obtain $\theta^{q+r}(0) h^{n-1,k-1}_{[r]}\in\sigma_k\left(\theta^{q+r-1}(0) h^{n-1,k-1}_{[r-1]}\right)$, thus  according to Eq. (\ref{Gamma-A-ge3}):
$h^{n,k}_{[qp^{n-1}+r]}\in\sigma_k\left( h^{n,k}_{[qp^{n-1}+r-1]}\right)$.

\smallbreak
~~~~(i.ii) Secondly, we prove that, for each $q\in [1,p-1]$, the last term of $C_{q-1}$ 
and the initial term of $C_{q}$ are also connected under $\sigma_k$. 
At first take $r=0$ in Eq. (\ref{Gamma-A-ge3}): 
it follows from $\theta^{p^{n-1}}=id_A$ that we have
$h^{n,k}_{[qp^{n-1}]}=\theta^{q}(0) h^{n-1,k-1}_{[0]}=\theta^{p^{n-1}+q}(0) h^{n-1,k-1}_{[0]}$.
Now take $r=p^{n-1}-1$ in Eq.~(\ref{Gamma-A-ge3}) , moreover substitute $q-1\in [0,p-2]$ to $q\in [1,p-1]$: we obtain
$h^{n,k}_{[qp^{n-1}-1]}=\theta^{q+p^{n-1}-2}(0) h^{n-1,k-1}_{[p^{n-1}-1]}$.
It follows from $p=|A|\ge 3$ that $\theta(0)\neq\theta^{-2}(0)$: since $\theta$ is one-to-one this implies
$\theta^{q+p^{n-1}}(0)\neq \theta^{q+p^{n-1}-2}(0)$, thus  $\theta^{q+p^{n-1}}(0)\in\sigma_1\left(\theta^{q+p^{n-1}-2}(0)\right)$.
Since by induction 
we have $ h^{n-1,k-1}_{[0]}=\sigma_{k-1}\left( h^{n-1,k-1}_{[p^{n-1}-1]}\right)$, we obtain
$h^{n,k}_{[qp^{n-1}]}\in\sigma_k\left(\theta^{q+p^{n-1}-2}(0) h^{n-1,k-1}_{[qp^{n-1}-1]}\right)$ that is,
 $h^{n,k}_{[qp^{n-1}]}\in\sigma_k\left( h^{n,k}_{[qp^{n-1}-1]}\right)$.

\smallbreak
~~~~(i.iii)
At last, we  prove that the first term of $C_0$ is connected under $\sigma_k$ with the last term of $C_{p-1}$.
In  Eq. (\ref{Gamma-A-ge3}), take  $q=0$ and $r=0$: we obtain $h^{n,k}_{[0]}=0h^{n-1,k-1}_{[0]}$.
Similarly, by setting $q=p-1$ and $r=p^{n-1}-1$,
we obtain $h^{n,k}_{[(p-1)p^{n-1}+p^{n-1}-1]}=\theta^{p^{n-1}+p-2}(0) h^{n-1,k-1}_{[p^{n-1}-1]}$,
thus $h^{n,k}_{[p^n-1]}=\theta^{-2}(0)h^{n-1,k-1}_{[p^{n-1}-1]}$.
Since $ h^{n-1,k-1}$  is a $\sigma_{k-1}$-Gray cycle over $A^{n-1}$, we have $ h^{n-1,k-1}_{[0]}\in\sigma_{k-1}\left( h^{n-1,k-1}_{[p^{n-1}-1]}\right)$.
In addition, it follows from $p\ge 3$, that $\theta^{-2}(0)\neq 0$, thus $0\in\sigma_1\left(\theta^{-2}(0)\right)$. 
We obtain
$ h^{n,k}_{[0]}
\in\sigma_k\left(\theta^{-2}(0)h^{n-1,k-1}_{[p^{n-1}-1]}\right)$, thus  $h^{n,k}_{[0]}\in\sigma_k\left( h^{n,k}_{[p^{n}-1]}\right)$ that is, the required property.

\smallbreak\noindent
(ii) Now, we prove that $ h^{n,k}$ satisfies Cond. \ref{ivb} that is, 
in its terms are pairwise different.
Let $i,i'\in [0,p^n-1]$ st. $h^{n,k}_{[i]}= h^{n,k}_{[i']}$ and consider the unique  $4$-tuple  of integers
$q,q'\in[0,p-1]$, $r,r'\in [0,p^{n-1}-1]$ st. $i=qp^{n-1}+r$ and $i'=q' p^{n-1}+r'$. 
According to  Eq. (\ref{Gamma-A-ge3}) we have $\theta^{q+r}(0)h^{n-1,k-1}_{[r]}=\theta^{q'+r'}(0)h^{n-1,k-1}_{[r']}$: since we have 
$\theta^{q+r}(0),\theta^{q'+r'}(0)\in A$, this implies $\theta^{q+r}(0)=\theta^{q'+r'}(0)$, whence we have $h^{n-1,k-1}_{[r]}=h^{n-1,k-1}_{[r']}$.
Since $h^{n-1,k-1}$ satisfies Cond.  \ref{ivc}, the second equation implies $r=r'$, whence
 the first one implies $\theta^{q}(0)=\theta^{q'}(0)$, thus  $q=q'\bmod p$. Since we have $q,q'\in [0,p-1]$ we obtain $q=q'$, thus  $i=i'$.

\smallbreak\noindent
(iii) At last, since the terms of $ h^{n,k}$ are pairwise different, we have:
$\left|\bigcup_{0\le i\le p^{n}-1}\{h^{n,k}_{[i]}\}\right|=p^{n}$,  thus: 
$\bigcup_{0\le i\le p^{n}-1}\{h^{n,k}_{[i]}\}=A^n$: this completes the proof.
\end{proof}
\section{An  alternative approach for computing the sequence $h^{n,k}$}
\label{A>3bis}
According to Eq. (\ref{Gamma-A-ge3}), given an integer pair $n,k$, the sequence $h^{n,k}$ can be computed by starting with $h^{n-1,k-1}$.
Clearly, in view of the full computation of  $h^{n,k}$, that type of approach leads to a recursive algorithm.
In the present section  we will provide some alternative method: actually  we will prove that,  for each $i\in [0,p^{n}-1]$, the word $h^{n,k}_{[i]}$ can be directly computed by starting with $h^{n,k}_{[i-1]}$.
Beforehand we  introduce  some complementary definitions and notations:

\smallbreak
-  Given an integer pair $a,b\in {\mathbb N}$, with $a\le b$, and a finite word sequence   $x=\left(x_{[i]}\right)_{a\le i\le b}$, for convenience we set  $x=x_{[a\cdot\cdot b]}$.

\smallbreak
- With the preceding notation, given a positive integer  $c$, we say that the sequence $x$ is $c$-{\it periodic}
if either we have $|x|=b-a+1\le c$, or the equation $x_{[i+c]}=x_{[i]}$ holds for every $i\in [a,b-c]$. 
In particular, wrt. finite sequence concatenation,  $x=y^q$ with $q>0$ implies $x$ being $|y|$-periodic.

\smallbreak
 - Given a non-negative integer $m$ with  $m\le \max\{|x_i|: i\in [a,b]\}$,
we set $A^{-m}x= \left(A^{-m}x_{[i]}\right)_{a\le i\le b}$.

\smallbreak
- At last, given a positive integer $m$, and given $w\in A^*$, with $|w|\ge m$, it is convenient to denote by ${\rm P}_m(w)$ the unique prefix of $w$ that belongs to $A^m$:
this notation can be extended in a straightforward way to any sequence of words in $A^mA^*$.
\subsection{A property involving periodicity}
Recall that we set $n_0=n-k+1$. With this notation the following prop. holds:
\begin{lem}
\label{h-power}
For every $j\in [n_0,n]$ the sequence $A^{j-n}h^{n,k}$ is  $p^{j}$-periodic. 
More precisely $A^{j-n}h^{n,k}$ is  a concatenation power of $h^{j,k-n+j}_{[0\cdot\cdot p^{j}-1]}$.
\end{lem}
\begin{proof}
We apply a top-down induction-based argument over  $j\in [n_0,n]$. 

\smallbreak
-- The base case corresponds to $j=n$. With this condition we have  $A^{j-n}h^{n,k}=h^{n,k}=h^{j,k-n+j}_{[0\cdot\cdot p^{j}-1]}$, thus trivially the prop.  holds.

\smallbreak
-- For the induction stage, we assume that  $A^{j-n}h^{n,k}$ is a concatenation power of the sequence $h^{j,k-n+j}_{[0\cdot\cdot p^{j}-1]}$.
With this condition,  the sequence $A^{j-n-1}h^{n,k}=A^{-1}\left(A^{j-n}h^{n,k}\right)$  is a concatenation power of $A^{-1}h^{j,k-n+j}_{[0\cdot\cdot p^{j}-1]}$.
 Let $i$ be an arbitrary integer   in $[0,p^{j}-1]$ and  let $q\in [0,p-1]$, $r\in [0,p^{j-1}-1]$ be the unique integer pair st. $i=qp^{j-1}+r$.
By substituting $j$ to $n$ and  $k-n+j$ to $k$ in Eq. (\ref{Gamma-A-ge3}), we obtain
$A^{-1}h^{j,k-n+j}_{[qp^{j-1}+r]}=h^{j-1,k-n+j-1}_{[r]}$.
 As a consequence,  we have
$A^{-1}h^{j,k-n+j}_{[qp^{j-1}\cdot\cdot(q+1)p^{j-1}-1]}=h^{j-1,k-n+j-1}_{[0\cdot\cdot p^{j-1}-1]}$ therefore, wrt.  sequence concatenation, we obtain
$\prod\limits_{0\leq q\le p-1} A^{-1}h^{j,k-n+j}_{[qp^{j-1}\cdot\cdot(q+1)p^{j-1}-1]}=\left(h^{j-1,k-n+j-1}_{[0\cdot\cdot p^{j-1}-1]}\right)^p$,
thus  $A^{-1}h^{j,k-n+j}_{[0\cdot\cdot p^{j}-1]}=\left( h^{j-1,k-n+j-1}_{[0\cdot\cdot p^{j-1}-1]} \right)^p$.  
Consequently, the sequence  $A^{j-n-1}h^{n,k}$,  which is a concatenation power of $A^{-1}h^{j,k-n+j}_{[0\cdot\cdot p^{j}-1]}$,  is a concatenation power of  
$h^{j-1,k-n+j-1}_{[0\cdot\cdot p^{j-1}-1]}$. 
\end{proof}
\subsection{Some map of the combinatorial structure of $h^{n,k}$}
For every $i\in [0,p^n-1]$ 
Eq.  (\ref{Gamma-A-ge3}) 
leads to  compute $h^{n,k}_{[i]}$ by applying  a series of one-character left-concatenations. 
On the other hand, with the convention over the character positions in words   the following equation holds:
\begin{eqnarray}
\label{decompose}
h^{n,k}_{[i]}=\left(h^{n,k}_{[i]}\right)_n\left(h^{n,k}_{[i]}\right)_{n-1}\cdots\left(h^{n,k}_{[i]}\right)_{1},
~~{\rm with~~}\left(h^{n,k}_{[i]}\right)_j\in A~~ (j\in [1,n]).
\end{eqnarray}
Our aim is to prove that the word $h^{n,k}_{[i]}$ can be directly computed starting from $h^{n,k}_{[i-1]}$. 
In order to do this, we are going to highlight a combinatorial structure common to all words $h^{j,k-n+j}_{[i]}$ $(j\in [n_0,n], i\in [0, p^{j}-1])$.
In what follows, we fix the integers $n$ and $k$. According to Eq. (\ref{decompose}), we set:
\begin{eqnarray}
\label{gamma-c0}
H^{[j]}_{[i]}=\left\{
\begin{array}{l}
\label{gamma-0}
\left(h^{n,k}_{[i]}\right)_j\in A~~{\rm if}~~j\in [n_0+1,n]\\
\label{gamma_1}
h^{n_0,1}_{[i]}=\left(h^{n,k}_{[i]}\right)_{n_0}\cdots \left(h^{n,k}_{[i]}\right)_1\in A^{n_0}~~{\rm if}~~j=n_0.
\end{array}
\right.
\end{eqnarray}
Let $H$ be the matrix with components $H^{[j]}_{[i]}$.
 The row index is  $i\in [1, p^{n}-1]$, the column index being $j\in [n_0+1,n]$. 
We emphasize on the fact that the row of index $i$ is:
\begin{eqnarray}
\label{the row H-i}
\left( H^{[n]}_{[i]}, H^{[n-1]}_{[i]},\cdots,H^{[j+1]}_{[i]}, H^{[j]}_{[i]},\cdots, H^{[n_0+1]}_{[i]},H^{[n_0]}_{[i]}   \right)=\nonumber\\
\left(\left(h^{n,k}_{[i]}\right)_n,\left(h^{n,k}_{[i]}\right)_{n-1},\cdots, \left(h^{n,k}_{[i]}\right)_{n_0+1},
\left(h^{n,k}_{[i]}\right)_{n_0}\cdots \left(h^{n,k}_{[i]}\right)_1\right).
\end{eqnarray}
Although at first glance the preceding notation may seem cumbersome, there are several reasons why we have adopted it:

\smallbreak
-- Firstly, the  notation is naturally connected to the computation process generated Eq.  (\ref{Gamma-A-ge3}).
In particular, it is  coherent with the right to left computation of $h^{n_0,1}$ using Algorithm (b) \cite[p.6]{K05}. 

\smallbreak
-- Secondly, with reference to the origins of Gray code topic, the word $h^{n,k}_{[i]}=H^{[n]}_{[i]}\cdots  H^{[n_0]}_{[i]}$ (or, equivalently, the polynomial in (\ref{the row H-i})) is actually the representation  in base $p=|A|$ of some non-negative  integer.

\smallbreak
-- Last and not least, regarding our experience concerning the present study, in  adopting  some alternative non-reversal representation, the results that we state below,  and their  proofs, would become much more difficult to read.
In particular  Eq.  (\ref{Gamma-A-ge3}), would be applied to more heavy indices, moreover the period of the columns in $H$ could not  take an  expression as simple as $p^j$  (see Lemma \ref{H-period}).

\medbreak
We introduce an additional notation: 
given $i\in[0,p^n-1]$ and $j\in [n_0,n]$,  we denote by $r(i,j)$ the unique integer in $[0,p^{j}-1]$ st. $i=r(i,j)\bmod p^j$. 
In particular, it follows from $i\in [0,p^n-1]$ that $r(i,n)=i$. 
As a consequence of Lemma \ref{h-power}, we obtain the following statement:
\begin{lem}
\label{H-period}
With the preceding notation, each of the following props. holds:

{\rm (i)}  For each $j\in [n_0,n]$ the sequence of words $H^{[j]}_{[0\cdot\cdot p^n-1]}$  is $p^{j}$-periodic.



{\rm (ii)}  For each $j\in [n_0+1,n]$, we have $H^{[j]}_{[i]}={\rm P}_1\left(h^{j,k-n+j}_{[r(i,j)]}\right)$. 
\end{lem}
\begin{proof}
(i) Firstly  assume  $j\in [n_0+1,n]$. According to Eq. (\ref{gamma-c0}),
in the matrix $H$ the component $H^{[j]}_{[i]}=\left(h^{n,k}_{[i]}\right)_j\in A$ is the initial character of the word $\left([H^{[n]}_{[i]}H^{[n-1]}_{[i]}\cdots  H^{[j+1]}_{[i]} \right)^{-1}h^{n,k}_{[i]}$, thus with the notation introduced above:
$H^{[j]}_{[i]}={\rm P}_1\left( A^{j-n}h^{n,k}_{[i]}\right)$.  According to Lemma \ref{h-power} the sequence $A^{j-n}h^{n,k}$ is  $p^{j}$-periodic, therefore $H^{[j]}_{[0\cdot\cdot p^n-1]}$ itself  is $p^{j}$-periodic.

\noindent
Now, we assume  $j=n_0$.  Once more according to Eq. (\ref{gamma-c0}),
we have:

  $H^{[n_0]}_{[i]}=h^{n_0,1}_{[i]}=\left(H^{[n]}_{[i]}H^{[n-1]}_{[i]}\cdots    H^{[n_0+1]}_{[i]}\right)^{-1}h^{n,k}_{[i]}\in A^{n-n_0}h^{n,k}_{[i]}$.

\noindent
Once more according to Lemma \ref{h-power}, the sequence $H^{[n_0]}_{[0\cdot\cdot p^n-1]}$ is $p^j$-periodic.
Consequently,  for each $j\in [n_0,n]$ the sequence of words $H^{[j]}_{[0\cdot\cdot p^n-1]}$  is $p^{j}$-periodic.

\smallbreak\noindent
(ii) Let $j\in [n_0+1,n_0]$. According to Lemma \ref{h-power}, the sequence  $A^{j-n}h^{n,k}$ is a concatenation power of
$h^{j,k-n+j}_{[0\cdot\cdot p^{j}-1]}$, hence 
 we have $A^{j-n}h^{n,k}_{[i]}=h^{j,k-n-j}_{[r(i,j)]}$ that is, 
by construction:
 $H^{[j]}_{[i]}={\rm P}_1\left(A^{j-n}h^{n,k}_{[i]}\right)={\rm P}_1\left(h^{j,k-n-j}_{[r(i,j)]}\right)$.
\end{proof}
\noindent
In addition  the following prop. holds:
\begin{lem}
\label{r(i,j)-i}
The condition $r(i,j)=0\bmod p^{j-1}$ is equivalent to $i=0\bmod p^{j-1}$. 
\end{lem}
\begin{proof} 
The cond. $i=r(i,j)\bmod p^j$ implies $i=r(i,j)\bmod p^{j-1}$. Consequently, $r(i,j)=0\bmod p^{j-1}$ implies $i=0\bmod p^{j-1}$.
Conversely, assume $i=mp^{j-1}$, with $m\in {\mathbb N}$. From the fact that $i=r(i,j)+m'p^{j}$, with $m'\in {\mathbb N}$, we obtain 
$r(i,j)=(m-m'p)p^{j-1}$, thus $r(i,j)=0\bmod p^{j-1}$.
\end{proof}
\subsection{An algorithmic interpretation}
On the basis of the above, an iteration-based method for computing the word $h^{n,k}_{[i]}$ can be drawn. 

-- We start by setting $H^{[n_0]}_{[0]}=h^{n_0,1}_{[0]}=0^{n_0}$.

-- Next,  for $i\in [0, p^{n_0}-1]$, according to Formula (\ref{gamma-c0}), the component $H^{[n_0]}_{[i]}$ is actually $h^{n_0,1}_{[i]}$, which can be generated starting from $h^{n_0,1}_{[i-1]}$ by construction.
 
-- In addition, for every $i\in [0,p^n-1]$, according to Lemma \ref{H-period},
 we have $H^{[n_0]}_{[i]}=h^{n_0,1}_{[r(i,n_0)]}$.
In other words,   wrt.  sequence concatenation, the column of index $n_0$ in the matrix $H$ is obtained by applying the following equation:
\begin{eqnarray}
\label{H-n-0}
H^{[n_0]}=\left(h^{n_0,1}\right)^{p^{n-n_0}}.
\end{eqnarray}
In what follows we introduce a permutation, namely
 $\omega$:

$\omega$ operates onto the  set $\left\{H^{[n_0]}_{[0]}, H^{[n_0]}_{[1]},\cdots,H^{[n_0]}_{[p^{n_0}-1]}\right\}$ as indicated in the following
\begin{eqnarray}
\label{CCD}
\omega\left(H^{[n_0]}_{[i]}\right)=\left\{
\begin{array}{l}
H^{[n_0]}_{[i+1]}~~{\rm if}~~i\in [0,p^{n_0}-2]\\
H^{[n_0]}_{[0]} ~~{\rm if}~~i=p^{n_0-1}.
\end{array}
\right.
\end{eqnarray}
Accordingly, for every $i\in [1,p^{n_0}-1]$ we have $H^{[n_0]}_{[i]}=h^{n_0,1}_{[i]}$. Note that, thanks to Algorithm (b) (see the preliminaries), $h^{n_0,1}_{[i]}$ can be generated from $h^{n_0,1}_{[i-1]}$ that is,
$H^{[n_0]}_{[i]}$ can be generated from $H^{[n_0]}_{[i-1]}$ . 
The following result is at the basis of an iterative algorithm for computing the whole matrix $H$
(recall that $\theta$ stands for the cyclic permutation $\left(0,\cdots,p-1\right)$):
\begin{proposit}\label{H}
Given $i\in [1,p^{n}-1]$, each of the following  equations hold:

{\rm (i)} $H^{[n_0]}_{[i]}=\omega\left(H^{[n_0]}_{[i-1]}\right)$.

{\rm (ii)} For every $j\in [n_0+1,n]$:
\begin{eqnarray}
 H^{[j]}_{[i]}=\left\{
\begin{array}{l}
\theta^2\left(H^{[j]}_{[i-1]}\right)~~{\rm if}~~i=0 \bmod p^{j-1}\\
\\
\theta\left(H^{[j]}_{[i-1]}\right)~~{\rm otherwise.}
\end{array}
\right.
\nonumber
\end{eqnarray}
\end{proposit}
\begin{proof}
(i) Firstly, assume $i\ne 0\bmod p^{n_0}$.  This condition implies $r(i,n_0)\in [1,p^{n_0}-1]$, 
hence  we have  $r(i-1,n_0)=r(i,n_0)-1\in [0,p^{n_0}-1]$. According to the  definition of $\omega$, this implies 
$H^{[n_0]}_{[r(i,n_0)]}=\omega\left(H^{[n_0]}_{[r(i-1,n_0)]}\right)$. According to the prop. (i) of Lemma \ref{H-period}, we obtain 
$H^{n_0,1}_{[i]}=H^{[n_0]}_{[r(i,n_0)]}=\omega\left(H^{[n_0]}_{[r(i-1,n_0)]}\right)=\omega\left(H^{[n_0]}_{[i-1]}\right)$.

\noindent
Now, we assume $i=0\bmod p^{n_0}$.
Once more according to the prop. (i) of Lemma \ref{H-period} we have 
$H^{[n_0]}_{[i]}=H^{[n_0]}_{[0]}$ and
$H^{[n_0]}_{[i-1]}=H^{[n_0]}_{[p^{n_0}-1]}$ (recall that, according to the condition of the present proposition we have $i\ge 1$). By the definition of $\omega$ (see (\ref{CCD}), we have
 $H^{[n_0]}_{[0]}=\omega\left(H^{[n_0]}_{[p^{n_0}-1]}\right)$, thus  $H^{[n_0]}_{[i]}=H^{[n_0]}_{[i-1]}$. 
This completes the proof of the prop. (i) of our proposition.

\medbreak
\noindent
(ii) Let $j\in [n_0+1,n]$, and  let $q, s$ be the unique integer pair st. $r(i,j)=qp^{j-1}+s$, with $0\le s\le  p^{j-1}-1$.
By substituting $r(i,j)$ to $i$, $j$ to $n$, and $k-n+j$ to $k$  in Eq. (\ref{Gamma-A-ge3}) ,we  obtain $h^{j,k-n+j}_{[r(i,j)]}=\theta^{q+s}(0)h^{j-1,k-n+j-1}_{[s]}$,  
whence $\theta^{q+s}(0)$ is the initial character of $h^{j,k-n+j}_{[r(i,j)]}$. According to the prop. (ii) of Lemma \ref{H-period}, we obtain
 $H^{[j]}_{[i]}=\theta^{q+s}(0)$. Now we examine the component $H^{[j]}_{[i-1]}$. 
Actually, according to the value of $s$ exactly one of the two following conds. occurs:

\medbreak
\noindent
{\underline {\it Cond. $s=0$}}\\
By definition this cond. is equivalent to  $r(i,j)=qp^{j-1}$.
It follows from $i\ge 1$ that $r(i-1,j)=r(i,j)-1=qp^{j-1}-1=(q-1)p^{j-1}+(p^{j-1}-1)$. Note that we have $p^{j-1}-1\in [0,p^{j-1}-1]$:
by substituting $r(i-1,j)$ to $i$, $j$ to $n$, and $k-n+j$ to $k$ in  Eq.  (\ref{Gamma-A-ge3}), we obtain
$h^{j,k+j-n}_{[r(i-1,j)]}=h^{j,k+j-n}_{[(q-1)p^{j-1}+(p^{j-1}-1)]}=
\theta^{(q-1)+(p^{j-1}-1)}(0)h^{j-1,k+j-n-1}_{[p^{j-1}-1]}$, whence $\theta^{q+p^{j-1}-2}(0)=\theta^{q-2}(0)$ is the initial  character of $h^{j,k-n+j}_{[r(i-1,j)]}$.
According to the prop.  (ii) of Lemma \ref{H-period}, we have  $H^{[j]}_{[i-1]}=\theta^{q-2}(0)$.
As indicated above, we have  $H^{[j]}_{[i]}=\theta^{q+s}(0)$: we obtain $H^{[j]}_{[i]}=\theta^{q}(0)=
\theta^2\left(\theta^{q-2}(0)\right)=\theta^2\left( H^{[j]}_{[i-1]}\right)$.

\smallbreak
\noindent
{\underline {\it Cond. $s>0$}}\\
With this cond. we have $r(i-1,j)=qp^{j-1}+(s-1)$, with  $0\le s-1\le p^{j-1}-1$.
Once more by substituting  in  Eq. (\ref{Gamma-A-ge3}) $r(i-1,j)$ to $i$, $j$ to $n$, and $k-n+j$ to $k$, we obtain $h^{j,k-n+j}_{r[(i-1,j)]}=\theta^{q+s-1}(0)h^{j-1,k-n+j-1}_{[s-1]}$,
whence $\theta^{q+s-1}(0)$ is the initial character of $h^{j,k-n+j}_{[r(i-1,j)]}$. 
According to the prop. (ii) of Lemma \ref{H-period}, this implies $H^{[j]}_{[i-1]}=\theta^{q+s-1}(0)$ that is, 
$H^{[j]}_{[i]}=\theta^{q+s}(0)=\theta\left(\theta^{q+s-1}(0)\right)=\theta\left(H^{[j]}_{[i-1]}\right)$.
\end{proof}
\noindent
In view of  Proposition \ref{H},  an iterative algorithm  to the sequence $h^{n,k}$, namely Algorithm 1,
can be drawn.
This algorithm computes row by row each component of the matrix $H$.
Such a computation actually makes use  of a unique generic row, namely:
  $${\cal H}=\left({\cal H}^{[n]},{\cal H}^{[n-1]},\cdots,{\cal H}^{[j]},{\cal H}^{[j-1]},\cdots, {\cal H}^{[n_0+1]},{\cal H}^{[n_0]}\right).$$
From this point of view, each time the counter $i$ is incremented, the generic row ${\cal H}$ is updated to 
$\left(H^{[n]}_{[i]},H^{[n-1]}_{[i]},\cdots,H^{[j]}_{[i]},H^{[j-1]}_{[i]},\cdots, H^{[n+0]}_{[i]},H^{[n_0]}_{[i]}\right)$.
Recall that, according to Eq. (\ref{H-n-0}) the following equalities  hold:
\begin{equation}
\label{H0}
H^{[n_0]}_{[0]}=h^{n_0,1}_{[0]}=0^{n_0},~~H^{[n_0]}_{[0\cdot\cdot p^{n_0}-1]}=h^{n_0,1},~~~~H^{[n_0]}_{[0\cdot\cdot p^{n}-1]}=\left(h^{n_0,1}\right)^{p^{n-n_0}}
\end{equation}
\begin{figure}[H]
\begin{center}
\includegraphics[width=15cm,height=8cm]{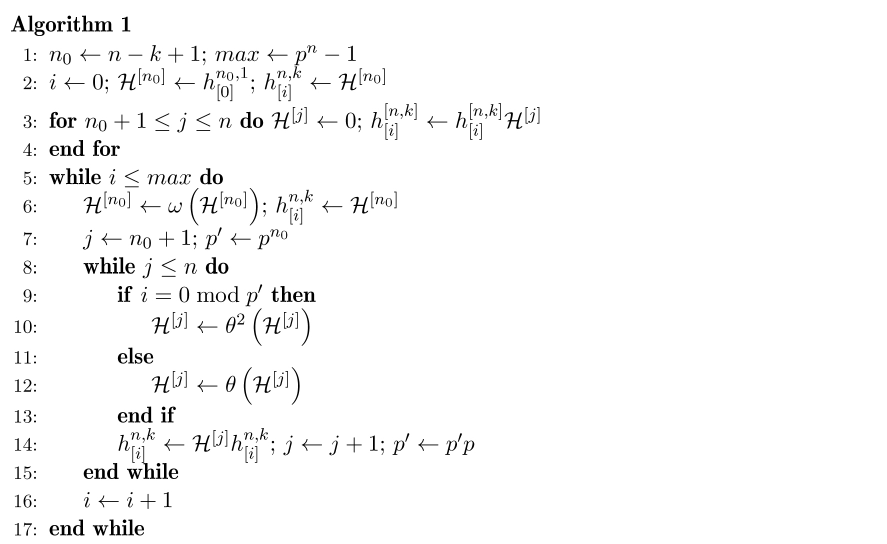}
\end{center}
\label{example-reg-subset-M}
\end{figure}
\medbreak\noindent
{\it A few comments on Algorithm 1}\\
-- The component  ${\cal H}^{[n_0]}$ is initialized to $H^{[n_0]}_{[0]}=h^{n_0,1}_{[0]}$. This is done by applying the process described in lines 2--4.

\noindent
-- Each time the counter $i$ is incremented, in the stage described in lines 6--16 Algorithm 1 computes from right to left the row  $\left(H^{[n]}_{[i]},\cdots, H^{[n_0+1]}_{[i]},H^{[n_0]}_{[i]}\right)$.
The term $h^{n,k}_{[i]}$ itself  is computed as the concatenation $H^{[n]}_{[i]}\cdots H^{[n_0+1]}_{[i]}\cdot H^{[n_0]}_{[i]}$. 
This  is done by updating the generic row ${\cal H}$: its  components are   computed by applying the formula from Proposition \ref{H}.
In addition, at line 14, after incrementation of the variable $j$,  $p^{j-1}$ is memorized in the variable $p'$, which took the initial value $p^{n_0}$ (line 7).

\noindent
-- The algorithm  stops when  the counter $i$ reaches the value {\it max}+1$=p^n$ (see l. 6).

\medbreak\noindent
{\it Questions related to complexity}\\
By construction, applying Algorithm 1 for computing the whole sequence $h^{n,k}$ that is, computing all the rows of the matrix $H$, requires at most $np^n$ insertions. 
In what follows, our goal is to improve such a bound.

Beforehand we note that, in any case, the alphabet $A$ and the permutation $\theta$ should be computed in some preprocessing stage. 

\smallbreak
-- On the one hand, for each incrementation of the counter $i$, updating  the generic sequence ${\cal H}=\left({\cal H}^{[n]},\cdots {\cal H}^{[n_0+1]}\right)$
 is done by applying the stage in lines 8--15:  there is a positive integer, say $\ell$,
st. applying that stage requires at most $\ell(n-n_0)\le \ell k$ insertions. 
When the counter $i$ reaches the value $max$,  ${\cal H}$ has been updated $p^{n}$ times, whence the total amount of corresponding operations  is at most $p^n\ell k$.

\smallbreak
-- On the other hand, in order to update the component  ${\cal H}^{[n_0]}$,  in l. 6. we need to apply the permutation $\omega$.
From this point of view,  there are actually two  strategies of implementation: 

\smallbreak
(a)  In the first approach,  for each value of $i$, in order to  compute the finite sequence $h^{n_0,1}=\left(H^{[n_0]}_{[i]}\right)_{0\le i\le p^{n_0-1}}$ 
 we apply the step (\ref{CCh}) from Algorithm (b), as  mentioned in the preliminaries.
 Actually, the right-most character flips each time, the second one flips every $p$ time, and so on:
classically,  that method
requires an amount of $p^{n_0}+p^{n_0-1}+\cdots +p\le p^{n_0}$ one-character substitutions.
For computing the whole column $H^{[n_0]}=\left(H^{[n_0]}_{[i]}\right)_{0\le i\le p^{n-1}}$, the total cost of the preceding operations  is bounded by $p^{n-n_0}\cdot p^{n_0}=p^n$.
Consequently, computing the whole sequence $h^{n,k}$ requires a total amount of operations bounded by $p^n\ell k+p^n=p^n(\ell k+1)$.
Note that in the computation of each term of  $h^{n_0,1}$ the amount of substitutions actually depends of the  value of the counter $i$ that is, the process cannot be loopless.
However, the amortized cost per operation
is $p^{-n}p^n\left(\ell k+1\right)=\ell k+1$ that is, $O(1)$. In other words,  with this strategy  Algorithm 1 runs in constant amortized-time wrt. $n$, with space linear in $n-n_0+n_0=n$.

\smallbreak
(b) The second approach consists in implementing in a preprocessing phase the sequence $h^{n_0,1}$ and the permutation $\omega$:
such an implementation requires  space $O(n_0p^{n_0})$ and, as indicated above,  a total amount of $O(p^{n_0})$ substitutions.
After that, in the processing phase, updating ${\cal H}^{n_0}$ will be performed by applying the result of lemma \ref{H-period}: this leads to a constant number of requests to $h^{n_0,1}$ and $\omega$ (see l. 6).
In other words,  updating the whole sequence ${\cal H}$ requires  constant time that is, 
with such a strategy  implementation,  Algorithm 1 is loopless and  requires space linear in $n_0p^{n_0}+n$.
%
\section{The case where $A$ is a binary alphabet, with $k$ odd}
\label{S4}
Let $A=\{0,1\}$ and $n\ge k$. With this condition, the cyclic permutation  $\theta$, which was introduced in Section \ref{Prelim}, is defined by $\theta(0)=1$ and $\theta(1)=0$.
Classically, this permutation can be extended into a one-to-one monoid homomorphism onto $A^*$:
in view of this, we set $\theta(\varepsilon)=\varepsilon$ and,  for any non-empty $n$-tuple of characters  $a_1,\cdots, a_n\in A$, $\theta(a_1\cdots a_n)=\theta(a_1)\cdots\theta(a_n)$.
Trivially, in the case where we have $n=k$, if a non-empty $\sigma_k$-Gray code exists over $X\subseteq A^n$, then  we have $X=\{x,\theta(x)\}$, for some $x\in A^n$.
In the sequel we assume $n\ge k+1$; with this condition  we will construct a pair of  peculiar $\sigma_k$-Gray cycles over $A^n$, namely $\gamma^{n,k}$ and $\rho^{n,k}$.
This will be done by induction over $k'$, the unique non-negative integer st.  $k=2k'+1$. Beforehand, we set $n_0=n-2k'=n-k+1$:
necessarily we have $n_0\ge 2$.

\smallbreak
-- For the base case, $\gamma^{n_0,1}$ and $\rho^{n_0,1}$ are computed by applying  some
reversal (resp., shift) over the sequence $g^{n_0,1}$, which was introduced  in  the cond. (\ref{CCg}) grom  Section \ref{Prelim}. We set : 
\begin{eqnarray}
\label{gamma-init}
\gamma^{n_0,1}_{[0]}=g^{n_0,1}_{[0]}~~{\rm and}~~
\gamma^{n_0,1}_{[i]}=g^{n_0,1}_{[2^{n_0}-i]}~~(1\le i\le 2^{n_0}-1).\\
\label{gamma-init1}
\rho^{n_0,1}_{[0]}=g^{n_0,1}_{[2^{n_0}-1]}~~{\rm and}~~
\rho^{n_0,1}_{[i]}=g^{n_0,1}_{[i-1]}~~(1\le i\le 2^{n_0}-1).
\end{eqnarray}
By construction,  $\gamma^{n_0,1}$ and $\rho^{n_0,1}$ are $\sigma_1$-Gray cycles over $A^{n_0}$. Moreover we have:
\begin{eqnarray}
\label{termes-initiaux}
\gamma^{n_0,1}_{[0]}=0^{n_0}~~{\rm and}~~\rho^{n_0,1}_{[0]}=10^{n_0-1}.\\
\label{termes-initiaux1}
\gamma^{n_0,1}_{[2^{n_0}-1]}=g^{n_0,1}_{[1]}=0^{n_0-1}1~~{\rm and}~~\rho^{n_0,1}_{[2^{n_0}-1]}=g^{n_0,1}_{[2^{n_0}-2]}=10^{n-2}1.
\end{eqnarray}
\begin{example}
\label{Exemple2}
Taking $n_0=2$, according to Eqs. \ref{gamma-init1}, the corresponding sequences $\gamma^{2,1}$ and $\rho^{n_0,1}$ are the following ones:
$$
\begin{array} {c}
g^{2,1}\\ \overbrace{~~}\\ 00\\ 01\\ 11\\ 10
\end{array}
~~~~~~~~~~~~~~~~~
\begin{array} {c}
\gamma^{2,1}\\ \overbrace{~~}\\ 00\\ 10\\ 11\\ 01
\end{array}
~~~~~~~~~~~~~~~~~
\begin{array} {c}
\rho^{2,1}\\ \overbrace{~~}\\ 10\\ 00\\ 01\\ 11
\end{array}
$$
\end{example}
\begin{example} 
\label{Ex3} 
For $n_0=3$ we obtain the following sequences:
$$
\begin{array} {c}
g^{3,1}\\ \overbrace{~~}\\ 000\\ 001\\ 011\\ 010\\ 110\\ 111\\ 101\\ 100
\end{array}
~~~~~~~~~~~~~~~~~
\begin{array} {c}
\gamma^{3,1}\\ \overbrace{~~}\\ 000\\ 100\\ 101\\ 111\\ 110\\ 010\\ 011\\ 001
\end{array}
~~~~~~~~~~~~~~~~~
\begin{array} {c}
\rho^{3,1}\\ \overbrace{~~}\\ 100\\ 000\\ 001\\ 011\\ 010\\ 110\\ 111\\ 101
\end{array}
$$
\end{example}

\smallbreak
-- In view of the induction step, we assume that we have computed  the $\sigma_k$-Gray cycles $\gamma^{n,k}$ and $\rho^{n,k}$. 
Note that we have $n+2=n_0+2(k'+1)=n_0+(k+2)-1$:
below we explain  the construction of  the two corresponding $2^{n+2}$-term sequences $\gamma^{n+2,k+2}$ and $\rho^{n+2,k+2}$. 
Given $i\in [0,2^{n+2}-1]$, let $q\in [0,3]$, $r\in [0,2^{n}-1]$ be the unique integer pair st. $i=q2^n+r$.
By assigning to $q$ the value  $0$ (resp., $1$, $2$, $3$), we state the corresponding Eq. (\ref{6}) (resp.,  Eqs. (\ref{7}),(\ref{8}),(\ref{9})):

\begin{subeqnarray}
\gamma^{n+2,k+2}_{[r]}=\theta^{r}(00)\gamma^{n,k}_{[r]}\label{6}\\
\gamma^{n+2,k+2}_{[2^n+r]}=\theta^{r}(01)\rho^{n,k}_{[r]} \label{7}\\
\gamma^{n+2,k+2}_{[2\cdot 2^n+r]}=\theta^{r}(11)\gamma^{n,k}_{[r]} \label{8}\\
\gamma^{n+2,k+2}_{[3\cdot 2^n+r]}=\theta^{r}(10)\rho^{n,k}_{[r]} \label{9}
\end{subeqnarray}
Similarly the sequence $\rho^{n+2,k+2}$ is computed by substituting,  in the preceding  Eqs., the $4$-tuple 
$(10,11,01,00)$ to $(00,01,11,10)$:

\begin{subeqnarray}
\rho^{n+2,k+2}_{[r]}=\theta^{r}(10)\gamma^{n,k}_{[r]} \label{10}\\
\label{H2} \rho^{n+2,k+2}_{[2^n+r]}=\theta^{r}(11)\rho^{n,k}_{[r]} \label{11}\\
\label{H3} \rho^{n+2,k+2}_{[2\cdot 2^n+r]}=\theta^{r}(01)\gamma^{n,k}_{[r]} \label{12}\\
\label{H4} \rho^{n+2,k+2}_{[3\cdot 2^n+r]}=\theta^{r}(00)\rho^{n,k}_{[r]}\cdot\label{13}
\end{subeqnarray}

\begin{example}
(Example \ref{Ex3} continued)
$\gamma^{5,3}$  is the concatenation,  in this order, of  the $4$  following  subsequences:
$$
\begin{array} {c}
~~~~~{\scriptstyle\gamma^{3,1}}\\ ~~~~\overbrace{}\\
 00~000\\ 11~100\\ 00~101\\ 11~111\\ 00~110\\ 11~010\\00~011\\ 11~001
\end{array}
~~~~~~~~~~~~~~~~~
\begin{array} {c}
~~~~~~~{\scriptstyle\rho^{3,1}}\\ ~~~~\overbrace{~~~~~ }\\
01~100\\ 10~000\\ 01~001\\ 10~011\\01~010\\ 10~110\\01~111\\ 10~101
\end{array}
~~~~~~~~~~~~~~~~~
\begin{array} {c}
~~~~~~~{\scriptstyle\gamma^{3,1}}\\ ~~~~\overbrace{~~~~~ }\\
11~000\\ 00~100\\ 11~101\\ 00~111\\ 11~110\\ 00~010\\11~011\\ 00~001
\end{array}
~~~~~~~~~~~~~~~~~
\begin{array} {c}
~~~~~~~{\scriptstyle\rho^{3,1}}\\ ~~~~\overbrace{~~~~~ }\\10~100\\01~000\\ 10~001\\ 01~011\\10~010\\ 01~110\\10~111\\ 01~101
\end{array}
$$
\end{example}
\begin{lem}
\label{pairwise-different}
Both the sequences $\gamma^{n,k}$,
$\rho^{n,k}$ satisfy each of the conditions \ref{iva}, \ref{ivc}.
\end{lem}
\begin{proof}
 Recall that we set $k=2k'+1$: we argue by induction over  $k'\ge 0$. 
The base case corresponds to  $k'=0$ that is,  $k=1$ and $n=n_0$: 
as indicated above, $\gamma^{n_0,1}$ and $\rho^{n_0,1}$ are $\sigma_1$-Gray cycles over $A^{n}$.
In view of  the induction step we assume that, for some $k'\ge 0$, both the sequences 
$\gamma^{n,k}$ and $\rho^{n,k}$
are $\sigma_k$-Gray cycles over $A^n$.

\smallbreak\noindent
(i) In order to prove that $\gamma^{n+2,k+2}$ satisfies Cond. \ref{ivc},
we consider an integer pair $i,i'\in [0,2^{n+2}-1]$ st. 
 $\gamma^{n+2,k+2}_{[i]}=\gamma^{n+2,k+2}_{[i']}$. Let
$q,q'\in [0,3]$, $r,r'\in [0,2^n-1]$ st. $i=q2^n+r$, $i'=q'2^n+r'$.
According to  Eqs. (\ref{6})--(\ref{9}) there are words $x,x'\in A^2$, $w,w'\in A^n$ st.
$\gamma^{n+2,k+2}_{[i]}=\theta^r(x)w$ and $\gamma^{n+2,k+2}_{[i']}=\theta^{r'}(x')w'$ that is,
$\theta^r(x)=\theta^{r'}(x')\in A^2$ and $w=w'$. By the definition of $\theta$, this implies either $x,x'\in\{00,11\}$ or $x,x'\in\{01,10\}$ that is, by construction,
either $q,q'\in\{0,2\}$, $x,x'\in\{00,11\}$, $w=\gamma^{n,k}_{[r]}=\gamma^{n,k}_{[r']}$, or 
$q,q'\in\{1,3\}$, $x,x'\in\{01,10\}$, $w=\rho^{n,k}_{[r]}=\rho^{n,k}_{[r']}$.
Since $\gamma^{n,k}$ and $\rho^{n,k}$ satisfies  Cond. \ref{ivc}, in any case we have  $r=r'$.
This implies  $\theta^r(x)=\theta^{r}(x')$, thus $x=x'$. 
With regard  to   Eqs. (\ref{6})--(\ref{9}), this corresponds to  $q=q'$, thus  $i=q2^n+r=q'2^n+r=i'$, therefore $\gamma^{n+2,k+2}$ satisfies Cond. \ref{ivc}.

\smallbreak\noindent
(ii) By substituting  $(10,11,01,00)$ to  $(00,01,11,10)$, according to (\ref{10})--(\ref{13}),
similar arguments lead to prove that $\rho^{n+2,k+2}_{[i]}=\rho^{n+2,k+2}_{[i']}$ implies $i=i'$ that is,  the sequence $\rho^{n+2,k+2}$ also satisfies Cond. \ref{ivc}.

\smallbreak\noindent
 (iii) Since $\gamma^{n+2,k+2}$ satisfies \ref{ivc}, we have
$\bigcup_{0\le i\le 2^{n+2}-1}\{\gamma^{n+2,k+2}_i\}=A^{n+2}$, hence 
our sequence satisfies \ref{iva}.
Similarly, since  $\rho^{n+2,k+2}$ satisfies Cond. \ref{ivc}
it  satisfies Cond. \ref{iva}.
\end{proof}
\noindent
 In order to prove that both the sequences $\gamma^{n,k}$ and $\rho^{n,k}$  satisfy Condition \ref{ivb}, beforehand we establish the following prop.:
\begin{lem}
\label{For-Gray-connections}
We have $\gamma^{n,k}_{[0]}\in\sigma_{k+1}\left(\rho^{n,k}_{[2^{n}-1]}\right)$ and 
$\rho^{n,k}_{[0]}\in\sigma_{k+1}\left(\gamma^{n,k}_{[2^{n}-1]}\right).$
\end{lem}
\begin{proof}
We argue by induction over $k'\ge 0$.

\smallbreak\noindent
-- The base case corresponds to $k'=0$, thus $k=1$ and $n=n_0$.
According to  the identity (\ref{termes-initiaux}) we have $\gamma^{n_0,1}_{[0]}=0^{n_0}\in\sigma_2\left(10^{n_0-2}1\right)$ thus $\gamma^{n_0,1}_{[0]}\in\sigma_2\left(\rho^{n,k}_{[2^{n}-1]}\right)$. 
Similarly, according to (\ref{termes-initiaux1})  we have $ \rho^{n_0,1}_{[0]}=10^{n_0-1}\in\sigma_2\left(0^{n_0-1}1\right)$ that is, $\rho^{n_0,1}_{[0]}\in \sigma_2\left(\gamma^{n_0,1}_{[2^{n_0}-1]}\right)$.

\smallbreak\noindent
-- For the induction step, we assume that, for some $k'\ge 0$, we have $\gamma^{n,k}_{[0]}\in\sigma_{k+1}\left(\rho^{n,k}_{[2^{n}-1]}\right)$ and $\rho^{n,k}_{[0]}\in\sigma_{k+1}\left(\gamma^{n,k}_{[2^{n}-1]}\right)$.

\smallbreak
(i) In Eq. (\ref{6}), by setting $r=0$ we obtain
$\gamma^{n+2,k+2}_{[0]}=00\gamma^{n,k}_{[0]}$, hence by induction:
$\gamma^{n+2,k+2}_{[0]} \in 00\sigma_{k+1}\left(\rho^{n,k}_{[2^{n}-1]}\right)\subseteq \sigma_{k+3}\left(11\rho^{n,k}_{[2^{n}-1]}\right)$.
By setting $r=2^{n}-1$ in  Eq. (\ref{13}), we obtain $\rho^{n+2,k+2}_{[2^{n+2}-1]}=11\rho^{n,k}_{[2^{n}-1]}$: this implies
$\gamma^{n+2,k+2}_{[0]}\in\sigma_{k+3}\left(\rho^{n+2,k+2}_{[2^{n+2}-1]}\right)$.

\smallbreak
(ii) Similarly, by setting $r=0$  in Eq.  (\ref{10}), and  by induction we have:
$\rho^{n+2,k+2}_{[0]}=10\gamma^{n,k}_{[0]}\in\sigma_{k+3}\left(01\rho^{n,k}_{[2^{n}-1]}\right)$.
By taking $r=2^n-1$ in  Eq.  (\ref{9}) we obtain $\gamma^{n+2,k+2}_{[2^{n+2}-1]}=01\rho^{n,k}_{[2^n-1]}$, 
therefore we have $\rho^{n+2,k+2}_{[0]}\in\sigma_{k+3}\left(\gamma^{n+2,k+2}_{[2^{n+2}-1]}\right)$.
\end{proof}
Since  Eqs. (\ref{6})--(\ref{13}) look alike, one may be tempted to compress them by substituting to them some unique generic Formula.
Based on our tests, such a formula needs to introduce at least two additional technical parameters, with tedious handling. 
In the proof of the following result, we have opted to report some case-by-case basis argumentation:
this has the advantage of making use of arguments which, although being similar, are actually easily legible.
\begin{proposit}
\label{Gamma-n-k-odd}
 Both  the sequences $\gamma^{n,k}$ and $\rho^{n,k}$
are $\sigma_k$-Gray cycles over $A^n$.
\end{proposit}
%
\begin{proof}
Once more we argue by induction over $k'\ge 0$.
Since $\gamma^{n_0,1}$ and $\rho^{n_0,1}$ are  $\sigma_1$-Gray cycles over $A^n$, the prop. holds for $k'=0$.
In view of  the induction stage, we assume that, for some $k'\ge 0$ both the sequences $\gamma^{n,k}$ and $\rho^{n,k}$ are  $\sigma_k$-Gray cycles over $A^n$.
According to  Lemma \ref{pairwise-different}, it remains to establish that $\gamma^{n+2,k+2}$ and $\rho^{n+2,k+2}$ satisfy Cond. \ref{ivb} that is:
\begin{eqnarray}
\label{EQ1}
(\forall q\in \{0,1,2,3\})(\forall r\in [1,2^n-1]) ~\gamma^{n+2,k+2}_{[q2^n+r]}\in\sigma_{k+2}\left(\gamma^{n+2,k+2}_{[q2^n+r-1]}\right);\\
\label{EQ2}
(\forall q\in \{1,2,3\})~\gamma^{n+2,k+2}_{[q2^{n}]}\in 
\sigma_{k+2}\left(\gamma^{n+2,k+2}_{[q2^{n}-1]}\right);\\
\label{EQ3}\gamma^{n+2,k+2}_{[0]}\in\sigma_{k+2}\left(\gamma^{n+2,k+2}_{[2^{n+2}-1]}\right).
\end{eqnarray}
\begin{eqnarray}
\label{EQ4}
(\forall q\in \{0,1,2,3\})(\forall r\in [1,2^n-1]) ~\rho^{n+2,k+2}_{[q2^n+r]}\in\sigma_{k+2}\left(\rho^{n+2,k+2}_{[q2^n+r-1]}\right);\\
\label{EQ5}
(\forall q\in \{1,2,3\})~\rho^{n+2,k+2}_{[q2^{n}]}\in 
\sigma_{k+2}\left(\rho^{n+2,k+2}_{[q2^{n}-1]}\right);\\
\label{EQ6}
\rho^{n+2,k+2}_{[0]}\in\sigma_{k+2}\left(\rho^{n+2,k+2}_{[2^{n+2}-1]}\right).
\end{eqnarray}

 \noindent \underline{{\it Condition} (\ref{EQ1})}

\smallbreak\noindent
(i) At first assume $q=0$.
According to Eq.  (\ref{6}) and since by induction $\gamma^{n,k}$ satisfies Cond. \ref{ivb},
we have  $\gamma^{n+2,k+2}_{[r]}= \theta^{r}\left(00\right)\gamma^{n,k}_{[r]}\in\theta^{r}(00)\sigma_k\left(\gamma^{n,k}_{[r-1]}\right)$,
thus $\gamma^{n+2,k+2}_{[r]}\in\sigma_{k+2}\left(\theta^{r-1}(00)\gamma^{n,k}_{[r-1]}\right)$. 
In Eq.  (\ref{6}), by substituting $r-1$ to $r$ (we have $0\le r-1<2^n-1$):
we obtain  $\gamma^{n+2,k+2}_{[r-1]}=\theta^{r-1}(00)\gamma^{n,k}_{[r-1]}$: this implies
$\gamma^{n+2,k+2}_{[r]}\in\sigma_{k+2}\left(\gamma^{n+2,k+2}_{[r-1]}\right)$.

\smallbreak\noindent
(ii)  Now assume   $q=1$.
According to Eq. (\ref{7}), and since by induction $\rho^{n,k}$ satisfies Cond. \ref{ivb},
we have  $\gamma^{n+2,k+2}_{[2^n+r]}= \theta^{r}(01)\rho^{n,k}_{[r]}\in \sigma_{k+2}\left(\theta^{r-1}(01)\rho^{n,k}_{[r-1]}\right)$.
In Eq. (\ref{7}), substitute $r-1$ to $r$: we obtain $\gamma^{n+2,k+2}_{[2^n+r-1]}= \theta^{r-1}(01)\rho^{n,k}_{[r-1]}$, therefore we have
$\gamma^{n+2,k+2}_{[2^n+r]}\in\sigma_{k+2}\left(\gamma^{n+2,k+2}_{[2^n+r-1]}\right)$.

\smallbreak\noindent
(iii) For $q=2$, we make use of arguments very similar to those applied in (i): 
according to Eq. (\ref{8}) and since by induction $\gamma^{n,k}$ satisfies Cond. \ref{ivb},
we have  $\gamma^{n+2,k+2}_{[2\cdot 2^n+r]}= \theta^{r}(11)\gamma^{n,k}_{[r]}\in\theta^{r}(11)\sigma_k\left(\gamma^{n,k}_{[r-1]}\right)\subseteq \sigma_{k+2}\left(\theta^{r-1}(11)\gamma^{n,k}_{[r-1]}\right)$.
Once more in Eq.  (\ref{8}), by substituting $r-1$ to $r$,
we obtain  $\gamma^{n+2,k+2}_{[r-1]}=\theta^{r-1}(11)\gamma^{n,k}_{[r-1]}$: this implies
$\gamma^{n+2,k+2}_{[r]}\in\sigma_{k+2}\left(\gamma^{n+2,k+2}_{[r-1]}\right)$.

\smallbreak\noindent
(iv)  Finally, with the condition  $q=3$, according to Eq. (\ref{9}), and since $\rho^{n,k}$ satisfies Cond. \ref{ivb},
we have  $\gamma^{n+2,k+2}_{[3\cdot 2^n+r]}= \theta^{r}(10)\rho^{n,k}_{[r]}\in \sigma_{k+2}\left(\theta^{r-1}(10)\rho^{n,k}_{[r-1]}\right)$.
In Eq. (\ref{9}), substitute $r-1$ to $r$: we obtain $\gamma^{n+2,k+2}_{[3\cdot 2^n+r-1]}= \theta^{r-1}(10)\rho^{n,k}_{[r-1]}$, therefore we have
$\gamma^{n+2,k+2}_{[3\cdot 2^n+r]}\in\sigma_{k+2}\left(\gamma^{n+2,k+2}_{[3\cdot 2^n+r-1]}\right)$.

\medskip
 \noindent
\underline{{\it Condition} (\ref{EQ2})}

\smallbreak\noindent
(i) Assume $q=1$ and take $r=0$ in  Eq. (\ref{7}): we obtain $\gamma^{n+2,k+2}_{[2^n]}=01\rho^{n,k}_{[0]}$. It follows from Lemma \ref{For-Gray-connections},
that $\gamma^{n+2,k+2}_{[2^n]}\in01\sigma_{k+1}\left(\gamma^{n,k}_{[2^n-1]}\right)$, thus  $\gamma^{n+2,k+2}_{[2^n]}\in\sigma_{k+2}\left(11\gamma^{n,k}_{[2^n-1]}\right)$.
\noindent
By taking $r=2^n-1$ in  (\ref{6}), 
we obtain:
$\gamma^{n+2,k+2}_{[2^n-1]}=\theta^{2^n-1}\left(00\right)\gamma^{n,k}_{[2^n-1]}=11\gamma^{n,k}_{[2^n-1]}$: this implies
$\gamma^{n+2,k+2}_{[2^n]}\in\sigma_{k+2}\left(\gamma^{n+2,k+2}_{[2^n-1]}\right)$.

\smallbreak\noindent
(ii) Now, assume  $q=2$, and 
set $r=0$ in Eq. (\ref{8}): we obtain  $\gamma^{n+2,k+2}_{[2\cdot2^n]}=11\gamma^{n,k}_{[0]}$.
According to Lemma \ref{For-Gray-connections} 
we have $\gamma^{n+2,k+2}_{[2\cdot2^n]}\in11\sigma_{k+1}\left(\rho^{n,k}_{[2^n-1]}\right)\subseteq \sigma_{k+2}\left(10\rho^{n,k}_{[2^n-1]}\right)$.
By taking $r=2^n-1$ in (\ref{7}), we obtain $\gamma^{n+2,k+2}_{[2^n+ 2^n-1)]}=10\rho^{n,k}_{[2^n-1]}$, which implies
$\gamma^{n+2,k+2}_{[2\cdot 2^n]}\in\sigma_{k+2}\left(\gamma^{n+2,k+2}_{[2\cdot2^n-1]}\right)$.\\

\smallbreak\noindent
(iii)  Finally, for  $q=3$,  we take $r=0$ in Eq.  (\ref{9}). Once more according to Lemma \ref{For-Gray-connections},
we have $\gamma^{n+2,k+2}_{[3\cdot 2^n]}=10\rho^{n,k}_{[0]}\in 10\sigma_{k+1}\left(\gamma^{n,k}_{[2^n-1]}\right)\subseteq\sigma_{k+2}\left(01\gamma^{n,k}_{[2^n-1]}\right)$.
On the other hand, by taking $r=2^n-1$ in Eq. (\ref{9}):
we obtain $\gamma^{n+2,k+2}_{[2\cdot 2^n+2^n-1]}=\theta^{2^n-1}\left(10\right)\gamma^{n,k}_{[2^n-1]}=01\gamma^{n,k}_{[2^n-1]}$ that is,
$\gamma^{n+2,k+2}_{[3\cdot 2^n]}\in\sigma_{k+2}\left(\gamma^{n+2,k+2}_{[3\cdot 2^n-1]}\right)$.

\bigbreak
 \noindent
\underline{{\it Condition} (\ref{EQ3})}
\smallbreak
\noindent
Take $r=0$  in  Eq. (\ref{6}). According to Lemma \ref{For-Gray-connections}, we have  $\gamma^{n+2,k+2}_{[0]}=00\gamma^{n,k}_{[0]}\in 00\sigma_{k+1}\left(\rho^{n,k}_{[2^n-1]}\right)\subseteq\sigma_{k+2}\left(01\rho^{n,k}_{[2^n-1]}\right)$. 
By taking $r=2^n-1$ in Eq. (\ref{9}) we obtain 
$\gamma^{n,k}_{[3\cdot 2^n+ 2^n-1]}=01\rho^{n,k}_{[2^n-1]}$, thus
$\gamma^{n+2,k+2}_{[0]}\in\sigma_{k+2}\left(\gamma^{n,k}_{[2^{n+2}-1]}\right)$.

\bigbreak
 \noindent
\underline{{\it Condition} (\ref{EQ4})-- (\ref{EQ6})}

\smallbreak
\noindent
 According to the structures of  Eqs. (\ref{10})--(\ref{13}), for proving these conditions 
the method  consists in substituting
the word $\rho^{n+2,k+2}_{[r]}$ to $\gamma^{n+2,k+2}_{[r]}$, the $4$-tuple $(10,11,01,00)$ to $(00,01,11,10)$, and 
Eq.   (\ref{10}) {\large (}resp.,  (\ref{11}),  (\ref{12}),  (\ref{13}){\large )} to Eq.  (\ref{6})  {\large (}resp.,  (\ref{7}),  (\ref{8}),  (\ref{9}){\large )}.
\end{proof}

\section{Condition $k$ odd: a non recursive method for computing $\gamma^{n,k}$}
\label{S4bis}
Recall that 
we set $k=2k'+1$, $n_0=n-2k'=n-k+1\ge 2$. 
Let $J=\bigcup_{0\le\ell\le k'} \{n_0+2\ell\}=\{n_0, n_0+2,\cdots n-2,n\}$.
As in Sect. \ref{A>3bis}, we will establish an eq. that allows to compute the word $\gamma^{n,k}_{[i]}$ starting from $\gamma^{n,k}_{[i-1]}$, for every $i\in [1,2^n-1]$.
Beforehand, it is convenient to summarize such an approach.
\smallbreak
-- At first, some combinatorial study is drawn: for each $j\in J$ we describe, in term of periodicity, the structure of the sequence $A^{j-n}\gamma^{n,k}$
 (Lemma \ref{gamma-power}).

-- In order to provide some map of the whole family of words  $\gamma^{n,k}_{[i]}$ ($1\le i\le 2^n-1$), a first matrix, namely $C=\left(C^{[j]}_{[i]}\right)_{0\le i\le 2^n-1, j\in J}$, is introduced. 
The components  of $C$ are words that can  be computed, on the one hand by applying Eqs.  (\ref{Equations-C}), (\ref{Equations-Cn}) 
(such formulas actually come from the preceding eqs. (\ref{6})--(\ref{13})), and
on the other hand, 
by applying Lemma \ref{gamma-power}. 

\smallbreak
-- Actually  the matrix $C$ cannot be directly computed through an iteration-based method. 
To remedy this situation,  a second matrix namely $Q$ is introduced. From
 this point of view, Eqs. (\ref{Equations-q-C}), (\ref{Equations-q-Cn}) provide precision on Eqs. (\ref{Equations-C}), (\ref{Equations-Cn}).

\smallbreak
--At this stage, we have gathered sufficient material to obtain a first computation formula, which is presented in Lemma \ref{C}.
Some precision: in the case where $i$ is not a multiple of $2^{j-2}$,
our formula  allows to compute  the matrix $\left(Q^{[j]}_{[i]},C^{[j]}_{[i]}\right)$ by directly starting from $\left(Q^{[j]}_{[i-1]},C^{[j]}_{[i-1]}\right)$.
 In the case where $i$ is a multiple of $2^{j-2}$,  the formula requires  to start the computation from  the pair $\left(Q^{[j]}_{[i-2^{j-2}-1]},C^{[j]}_{[i-2^{j-2}-1]}\right)$:
 we have not yet achieved our goal.

\smallbreak
-- Furthermore, Proposition \ref{CC} sets a second formula:
in any case, it allows the computation of  $\left(Q^{[j]}_{[i]},C^{[j]}_{[i]}\right)$
by directly starting with  the pair $\left(Q^{[j]}_{[i-1]},C^{[j]}_{[i-1]}\right)$.

\smallbreak
-- Regarding the implementation of the above method, some pseudo-code is provided in Algorithm 2.
\subsection{A property involving periodicity}
We start by establishing the following result:
\begin{lem}
\label{gamma-power}
Wrt.  the sequence concatenation, for every $j\in J$ the sequence
$A^{j-n}\gamma^{n,k}$ is  $2^{j+1}$-periodic.
More precisely, given $j\in J\setminus {n}$,   $A^{j-n}\gamma^{n,k}$ is a power of the sequences concatenation  $\left(\gamma^{j,k-n+j}_{[0\cdot\cdot2^{j}-1]},\rho^{j,k-n+j}_{[0\cdot\cdot2^{j}-1]}\right)$.
\end{lem}
\begin{proof}
With the condition $j=n$, trivially the sequence $A^{j-n}\gamma^{n,k}=\gamma^{n,k}$ is $2^{j+1}$-periodic.
For $j\in J\setminus \{n\}$, by making use of a top-down induction-based argument over  $j$, we prove that $A^{j-n}\gamma^{n,k}$ is a concatenation power of  $\gamma^{j,k-n+j}_{[0\cdot\cdot2^{j}-1]}\rho^{j,k-n+j}_{[0\cdot\cdot2^{j}-1]}$.

\smallbreak
-- The base case corresponds to $j=n-2$. Let $i\in [0,2^{n}-1]$, $q\in [0,3]$, and $r\in [0, 2^{n-2}-1]$ st. $i=q2^{n-2}+r$. 
By substituting $n-2$ to $n$ and $k-2$ to $k$ in  Eqs.  (\ref{6})--(\ref{13}), we obtain the following identities: 
\begin{eqnarray} 
A^{-2}\gamma^{n,k}_{[r]}=\gamma^{n-2,k-2}_{[r]}\nonumber \\
A^{-2}\gamma^{n,k}_{[2^{n-2}+r]}=\rho^{n-2,k-2}_{[r]}\nonumber\\
A^{-2}\gamma^{n,k}_{[2\cdot 2^{n-2}+r]}=\gamma^{n-2,k-2}_{[r]}\nonumber\\
A^{-2}\gamma^{n,k}_{[3 \cdot 2^{n-2}+r]}=\rho^{n-2,k-2}_{[r]}.\nonumber
\end{eqnarray}
As a consequence, regarding sequences of words, each of the following eqs. holds:
 \begin{eqnarray}
A^{-2}\gamma^{n,k}_{[0\cdot\cdot2^{n-2}-1]}=\gamma^{n-2,k-2}_{[0\cdot\cdot2^{n-2}-1]}\nonumber\\
A^{-2}\gamma^{n,k}_{[2^{n-2}\cdot\cdot2\cdot 2^{n-2}-1]}=\rho^{n-2,k-2}_{[0\cdot\cdot2^{n-2}-1]}\nonumber\\
 A^{-2}\gamma^{n,k}_{[2\cdot 2^{n-2}\cdot\cdot3\cdot 2^{n-2}-1]}=\gamma^{n-2,k-2}_{[0\cdot\cdot2^{n-2}-1]}\nonumber\\
 A^{-2}\gamma^{n,k}_{[3\cdot 2^{n-2}\cdot\cdot2^{n}-1]}=\rho^{n-2,k-2}_{[0\cdot\cdot2^{n-2}-1]}.\nonumber
\end{eqnarray}
Wrt. sequence  concatenation, this implies:

 $A^{-2}\gamma^{n,k}=A^{-2}\gamma^{n,k}_{[0\cdot\cdot2^{n}-1]}=\left(\gamma^{n-2,k-2}_{[0\cdot\cdot2^{n-2}-1]},\rho^{n-2,k-2}_{[0\cdot\cdot2^{n-2}-1]}\right)^{2}$.

\noindent
Since the length of each of the sequences $\gamma^{n-2,k-2}_{[0\cdot\cdot2^{n-2}-1]},\rho^{n-2,k-2}_{[0\cdot\cdot2^{n-2}-1]}$ is $2^{n-2}=2^j$, the prop. of Lemma \ref{gamma-power} holds. 

\bigbreak
-- For the induction stage, we assume that, 
for some $j\in J\setminus \{n_0,n\}$, the sequence $A^{j-n}\gamma^{n,k}$ is a concatenation power of $\left(\gamma^{j,k-n+j}_{[0\cdot\cdot2^{j}-1]},\rho^{j,k-n+j}_{[0\cdot\cdot2^{j}-1]}\right)$.
With this condition, the sequence 
  $A^{j-n-2}\gamma^{n,k}=A^{-2}\left(A^{j-n}\gamma^{n,k}\right)$, for its part, is a power of $A^{-2}\left(\gamma^{j,k-n+j}_{0\cdot\cdot2^{j}-1},\rho^{j,k-n+j}_{0\cdot\cdot2^{j}-1}\right)$.
Let $i\in [0,2^j-1]$, $q\in [0,3]$, and $r\in [0,2^{j-2}-1]$ st. $i=q2^{j-2}+r$.
By substituting $j$ to $n+2$ and $k-n+j$ to $k+2$ in 
Eqs. (\ref{6})--(\ref{13}), we obtain:
\begin{eqnarray}
~~~~~~~~ A^{-2}\gamma^{j,k-n+j}_{[r]}=\gamma^{j-2,k-n+j-2}_{[r]}\nonumber\\
~~A^{-2}\gamma^{j,k-n+j}_{[2^{j-2}+r]}=\rho^{j-2,k-n+j-2}_{[r]}\nonumber\\
 A^{-2}\gamma^{j,k-n+j}_{[2\cdot 2^{j-2}+r]}=\gamma^{j-2,k-n+j-2}_{[r]}\nonumber\\
 ~~~~~A^{-2}\gamma^{j,k-n+j}_{[3\cdot 2^{j-2}+r]}=\rho^{j-2,k-n+j-2}_{[r]}.\nonumber
\end{eqnarray}
Therefore, the following equations hold:
 \begin{eqnarray}
~~~~~~~~ A^{-2}\gamma^{j,k-n+j}_{[0\cdot\cdot2^{j-2}-1]}=\gamma^{j-2,k-n+j-2}_{[0\cdot\cdot2^{j-2}-1]}\nonumber\\
~~A^{-2}\gamma^{j,k-n+j}_{[2^{j-2}\cdot\cdot2\cdot 2^{j-2}-1]}=\rho^{j-2,k-n+j-2}_{[0\cdot\cdot2^{j-2}-1]}\nonumber\\
 A^{-2}\gamma^{j,k-n+j}_{[2\cdot 2^{j-2}\cdot\cdot3\cdot 2^{j-2}-1]}=\gamma^{j-2,k-n+j-2}_{[0\cdot\cdot2^{j-2}-1]}\nonumber\\
 ~~~~~A^{-2}\gamma^{j,k-n+j}_{[3\cdot 2^{j-2}\cdot\cdot2^{j}-1]}=\rho^{j-2,k-n+j-2}_{[0\cdot\cdot2^{j-2}-1]}.\nonumber
\end{eqnarray}
This implies $A^{-2}\gamma^{j,k-n+j}_{[0\cdot\cdot2^{j-2}-1]}=\left(\gamma^{j-2,k-n+j-2}_{[0\cdot\cdot2^{j-2}-1]},\rho^{j-2,k-n+j-2}_{[0\cdot\cdot2^{j-2}-1]}\right)^2$.\\
Similarly, we have:
 \begin{eqnarray}
~~~~~~~~ A^{-2}\rho^{j,k-n+j}_{[0\cdot\cdot2^{j-2}-1]}=\gamma^{j-2,k-n+j-2}_{[0\cdot\cdot2^{j-2}-1]}\nonumber\\
~~A^{-2}\rho^{j,k-n+j}_{[2^{j-2}\cdot\cdot2\cdot 2^{j-2}-1]}=\rho^{j-2,k-n+j-2}_{[0\cdot\cdot2^{j-2}-1]}\nonumber\\
 A^{-2}\rho^{j,k-n+j}_{[2\cdot 2^{j-2}\cdot\cdot3\cdot 2^{j-2}-1]}=\gamma^{j-2,k-n+j-2}_{[0\cdot\cdot2^{j-2}-1]}\nonumber\\
 ~~~~~A^{-2}\rho^{j,k-n+j}_{[3\cdot 2^{j-2}\cdot\cdot2^{j}-1]}=\rho^{j-2,k-n+j-2}_{[0\cdot\cdot2^{j-2}-1]}\nonumber
\end{eqnarray}
Therefore we have $A^{-2}\rho^{j,k-n+j}_{[0\cdot\cdot2^{j-2}-1]}=\left(\gamma^{j-2,k-n+j-2}_{[0\cdot\cdot2^{j-2}-1]},\rho^{j-2,k-n+j-2}_{[0\cdot\cdot2^{j-2}-1]}\right)^2$.
We obtain:
 \begin{eqnarray}
A^{-2}\left(\gamma^{j,k-n+j}_{[0\cdot\cdot2^{j}-1]},\rho^{j-n,k-n+j}_{[0\cdot\cdot2^{j}-1]}\right)=\left( \gamma^{j-2,k-n+j-2}_{[0\cdot\cdot2^{j-2}-1]},\rho^{j-2,k-n+j-2}_{[0\cdot\cdot2^{j-2}-1]} \right)^4.\nonumber
 \end{eqnarray}
Consequently,  since it is a concatenation power of  $A^{-2}\left(\gamma^{j,k-n+j}_{[0\cdot\cdot2^{j}-1]},\rho^{j,k-n+j}_{[0\cdot\cdot2^{j}-1]}\right)$,
 the sequence  $A^{j-n-2}\gamma^{n,k}$  is a concatenation  power of the sequence:

  $\left(\gamma^{j-2,k-n+j-2}_{[0\cdot\cdot2^{j-2}-1]},\rho^{j-2,k-n+j-2}_{[0\cdot\cdot2^{j-2}-1]}\right)$.
Since the length of this last sequence  is $2\cdot 2^{j-2}=2^{(j-2)+1}$, the sequence $A^{j-n-2}\gamma^{n,k}$ itself  has period  $2^{(j-2)+1}$ that is,  the prop. of Lemma \ref{gamma-power} also holds for $j-2$.
This completes the proof.
\end{proof}
\subsection{Mapping the structure of $\gamma^{n,k}$}
\label{Matrix-C}
 Eqs. (\ref{6})--(\ref{13}) leads to compute $\gamma^{n,k}$ by recursively applying  a series of left-concatenation by words in $A^2$.
As previously announced, in the spirit of Sect. \ref{A>3bis} we introduce a matrix, namely $C$.
The row  index is $i\in [0, 2^n-1]$, the column index being $j\in J=\{n,n-2,\cdots, n_0+2,n_0\}$. 
More precisely, given $i\in [0,2^n-1]$ we set:
\begin{eqnarray}
\label{gamma-prod-C}
\gamma^{n,k}_{[i]}=C^{[n]}_{[i]}C^{[n-2]}_{[i]}\cdots C^{[j+2]}_{[i]}  C^{[j]}_{[i]}\cdots C^{[n_0+2]}_{[i]}C^{[n_0]}_{[i]},\\
{\rm with}~C^{[n_0]}_{[i]}\in A^{n_0}, ~{\rm and}~ C^{[j]}_{[i]}\in A^2~{\rm for~every~}j\in J\setminus\{ n_0\}.\nonumber
\end{eqnarray}
The reasons that can be  invoked for adopting such a reverse-order notation are the same that for the matrix $H$ from Sect. \ref{A>3bis}.
The row of index is $i\in [0, 2^n-1]$ is:  
\begin{eqnarray}
\left(C^{[n]}_{[i]},C^{[n-2]}_{[i]},\cdots C^{[j+2]}_{[i]},  C^{[j]}_{[i]}, \cdots, C^{[n_0+2]}_{[i]},C^{[n_0]}_{[i]}\right)
\end{eqnarray}
In addition,  we denote by $\mu(i,j)$ the unique integer in $[0,2^{j+1}-1]$ st. $i=\mu(i,j)\bmod 2^{j+1}$. 
As a consequence of Lemma \ref{gamma-power}, we obtain the following result, which is  the counterpart of  Lemma \ref{H-period} from Sect. \ref{A>3bis}:
\begin{lem}
\label{C-period}
With the preceding notation each of the following props. holds:

\smallbreak
 {\rm (i)}  For every $j\in J\setminus\{n_0,n\}$, we have $C^{[j]}_{[0\cdot\cdot2^{n}-1]}={\rm P}_2\left(A^{j-n}\gamma^{n,k}\right)$.
In addition the sequence 
$C^{[j]}_{[0\cdot\cdot2^n-1]}$ is a concatenation power of $C^{[j]}_{[0\cdot\cdot2^{j+1}-1]}$.

\smallbreak
{\rm (ii)} For every $i\in [0,2^n-1]$ we have:

~~~~~~~~~~~~~~~
$
C^{[n_0]}_{[i]}=
\begin{cases}
C^{[n_0]}_{[i]}= \gamma^{n_0,1}_{[\mu(i,n_0)]}~~{\it if}~~\mu(i,n_0)\in [0,2^{n_0}-1]\\
C^{[n_0]}_{[i]}= \rho^{n_0,1}_{[\mu(i,n_0)-2^{n_0}]}~~{\it if}~~\mu(i,n_0)\in [2^{n_0},2^{n_0+1}-1].
\end{cases}
$

\smallbreak

{\rm (iii)}   For every index pair $i\in [0,2^n-1]$, $j\in J\setminus\{n_0\}$, we have:

~~~~~~~~~~~~~~~
$
C^{[j]}_{[i]}=
\begin{cases}
{\rm P}_2\left(\gamma^{j,n-j+k}_{[\mu(i,j)]}\right)~~{\it if}~~\mu(i,j)\in [0,2^j-1]\\
{\rm P}_2\left(\rho^{j,n-j+k}_{[\mu(i,j)]-2^j}\right)~~{\it if}~~\mu(i,j)\in [2^j,2^{j+1}-1].
\end{cases}
$
\end{lem}
\begin{proof}
(i) Let $j\in J\setminus\{n_0,n\}$ and $i\in [0,2^{n-1}]$.
According to Eq. (\ref{gamma-prod-C}) we have 
$C^{[j]}_{[i]}={\rm P}_2\left(\left(C^{[n]}_{[i]}\cdots C^{[j+2)]}_{[i]}\right)^{-1}\gamma^{n,k}_{[i]}\right)$.

\noindent
Since the sequence
 $\left(C^{[n]}_{[i]},C^{[n-2]}_{[i]},\cdots, C^{[n-(n-j-2])}_{[i]}\right)$ has length $\frac{n-j-2}{2}+1$, 
we have
$\left | C^{[n]}_{[i]}\cdots C^{[j+2])}_{[i]}\right|=2\left(\frac{n-j-2}{2}+1\right)=n-j$, thus $C^{[j]}_{[i]}={\rm P}_2\left(A^{j-n}\gamma^{n,k}_{[i]}\right)$.
As a consequence, we obtain $ C^{[j]}_{[0\cdot\cdot 2^n-1]}
={\rm P}_2\left(A^{j-n}\gamma^{n,k}\right)$.

\noindent
As a direct consequence, we have   $C^{[j]}_{[0\cdots 2^{j+1}-1]}={\rm P}_2\left(A^{j-n}\gamma^{n,k}_{[0\cdot\cdot 2^{j+1}-1]}\right)$. 
According to Lemma \ref{gamma-power}, 
 the sequence $A^{j-n}\gamma^{n,k}$ is  a power of  $\left(\gamma^{j,k-n+j}_{[0\cdot\cdot2^{j}-1]},\rho^{j,k-n+j}_{[0\cdot\cdot2^{j}-1]}\right)$, whence
$ C^{[j]}_{[0\cdot\cdot 2^n-1]}={\rm P}_2\left(A^{j-n}\gamma^{n,k}\right)$ is a  power of ${\rm P}_2\left(\gamma^{j,k-n+j}_{[0\cdot\cdot2^{j}-1]},\rho^{j,k-n+j}_{[0\cdot\cdot2^{j}-1]}\right)$.
In addition, since the sequence $A^{j-n}\gamma^{n,k}_{[0\cdot\cdot 2^{j+1}-1]}$ has length $2^{j+1}$, we have 
$A^{j-n}\gamma^{n,k}_{[0\cdot\cdot 2^{j+1}-1]}=\left(\gamma^{j,k-n+j}_{[0\cdot\cdot 2^{j}-1]},\rho^{j,k-n+j}_{[0\cdot\cdot 2^{j}-1]}\right)$.

\noindent
This implies
$C^{[j]}_{[0\cdots 2^{j+1}-1]}={\rm P}_2\left(A^{j-n}\gamma^{n,k}_{[0\cdot\cdot 2^{j+1}-1]}  \right)={\rm P}_2\left(\gamma^{j,k-n+j}_{[0\cdot\cdot 2^{j}-1]},\rho^{j,k-n+j}_{[0\cdot\cdot 2^{j}-1]}\right) $.
Consequently,  ${\rm P}_2\left(A^{j-n}\gamma^{n,k}\right)$ is a power of
${\rm P}_2\left(\gamma^{j,k-n+j}_{[0\cdot\cdot 2^{j}-1]},\rho^{j,k-n+j}_{[0\cdot\cdot 2^{j}-1]}\right)$ that is, $ C^{[j]}_{[0\cdot\cdot 2^n-1]}$ is a concatenation power of $C^{[j]}_{[0\cdot\cdot 2^{j+1}-1]}$. 
This completes the proof of prop. (i). 

\smallbreak\noindent
(ii) According to Eq. (\ref{gamma-prod-C}), we have $C^{[n_0]}_{[0\cdots 2^{n_0+1}-1]}=\left(\gamma^{n_0,1}_{[0\cdots 2^{n_0}-1]},\rho^{n_0,1}_{[0\cdots 2^{n_0}-1]}\right)$.
By   taking $j=n_0\in J$ in the statement of Lemma \ref{gamma-power}, we observe that the sequence $A^{n_0-n}\gamma^{n,k}$ is  a power of $\left(\gamma^{n_0,1}_{[0\cdot\cdot 2^{n_0}-1]},\rho^{n_0,1}_{[0\cdot\cdot 2^{n_0}-1]}\right)$
(we have $k-n+n_0=1)$. 
Consequently, the condition  $\mu(i,n_0)\in [0,2^{n_0}-1]$, implies   $C^{[n_0]}_{[i]}= \gamma^{n_0,1}_{[\mu(i,n_0)]}$.
Similarly, the condition  $\mu(i,n_0)\in [2^{n_0},2^{n_0+1}-1]$, implies   $C^{[n_0]}_{[i]}= \rho^{n_0,1}_{[\mu(i,n_0)-2^{n_0}]}$.
This establishes  the prop. (ii) of Lemma \ref{C-period}.

\smallbreak\noindent
(iii) Let $i\in [0,2^n-1]$. 
According to the  prop. (i) of our lemma, $C^{[j]}_{[0\cdot\cdot2^{n}-1]}$ is  $2^{j+1}$-periodic. By the definition of $\mu(i,j)$ this implies $C^{[j]}_{[i]}=C^{[j]}_{[\mu(i,j)]}$.
In addition, $C^{[j]}_{[\mu(i,j)]}$ is the term of index $\mu(i,j)$ in the sequence ${\rm P}_2\left(\gamma^{j,k-n+j}_{[0\cdot\cdot2^{j}-1]}\rho^{j,k-n+j}_{[0\cdot\cdot2^{j}-1]}\right)$, hence the prop. (iii) holds.
\end{proof}

\bigbreak\noindent
According to Lemma \ref{C-period},  
given $i\in [0,2^{n_0+1}-1]$ the following equations hold :
\begin{eqnarray}
\label{Eq-C}
 C^{[n_0]}_{[i]}=\left\{ 
\begin{array}{l}
\gamma^{n_0,1}_{[i]}~~{\rm if}~~~i\in [0, 2^{n_0}-1]\\
\rho^{n_0,1}_{[i]} ~~{\rm if}~~~i\in [2^{n_0}, 2^{n_0+1}-1]
\end{array}
\right.
\end{eqnarray}
In particular, we have  $C^{[n_0]}_{[0]}=\gamma^{n_0,1}_{[0]}$.

\smallbreak\noindent
Let $i\in [0,2^n-1]$, $j\in J\setminus\{n,n_0\}$, and let
 $q'$, $r'$ be  the  unique integer pair st. $\mu(i,j)=q'2^{j-2}+r'$, with $r'\in [0,2^{j-2}-1]$.
It follows from $\mu(i,j)\in [0,2^{j+1}-1]$ that we have $q'\in [0,7]$, furthermore according to Eqs. (\ref{6})--(\ref{13}) each of the following identities holds:

\begin{equation}
\label{Equations-C1}
\begin{array} {c} 
~~ \gamma^{j,k-n+j}_{[r']}=\theta^{r'}(00)\gamma^{j-2,k-n+j-2}_{[r']}\\
~~ \gamma^{j,k-n+j}_{[2^{j-2}+r']}=\theta^{r'}(01)\rho^{j-2,k-n+j-2}_{[r']}\\
\gamma^{j,k-n+j}_{[2.2^{j-2}+r']}=\theta^{r'}(11)\gamma^{j-2,k-n+j-2}_{[r']}\\
\gamma^{j,k-n+j}_{[3.2^{j-2}+r']}=\theta^{r'}(10)\rho^{j-2,k-n+j-2}_{[r']}\\
\rho^{j,k-n+j}_{[4\cdot2^{j-2}+r']}=\theta^{r'}(10)\gamma^{j-2,k-n+j-2}_{[r']}\\
\rho^{j,k-n+j}_{[5\cdot 2^{j-2}+r']}=\theta^{r'}(11)\rho^{j-2,k-n+j-2}_{[r']}\\
\rho^{j,k-n+j}_{[6.2^{j-2}+r']}=\theta^{r'}(01)\gamma^{j-2,k-n+j-2}_{[r']}\\
\rho^{j,k-n+j}_{[7.2^{j-2}+r']}=\theta^{r'}(00)\rho^{j-2,k-n+j-2}_{[r']}.
\end{array}
\end{equation}
According to the prop. (iii) of  Lemma \ref{C-period}, we obtain the following equations:
\begin{equation}
\label{Equations-C}
\begin{array} {c} 
~~~~~~~~C^{[j]}_{[r']}=\theta^{r'}(00);\\
~~C^{[j]}_{[2^{j-2}+r']}=\theta^{r'}(01); \\
C^{[j]}_{[2\cdot 2^{j-2}+r']}=\theta^{r'}(11); \\
C^{[j]}_{[3\cdot 2^{j-2}+r']}=\theta^{r'}(10);\\
C^{[j]}_{[4\cdot 2^{j-2}+r']}=\theta^{r'}(10); \\
C^{[j]}_{[5\cdot 2^{j-2}+r']}=\theta^{r'}(11); \\
C^{[j]}_{[6\cdot 2^{j-2}+r']}=\theta^{r'}(01); \\
C^{[j]}_{[7\cdot 2^{j-2}+r']}=\theta^{r'}(00).
\end{array}
\end{equation}
In addition, for $j=n$,  the following identities holds:
\begin{equation}
\label{Equations-Cn}
\begin{array} {c} 
~~~~~~~~C^{[n]}_{[r']}=\theta^{r'}(00);\\
~~C^{[n]}_{[2^{n-2}+r']}=\theta^{r'}(01); \\
C^{[n]}_{[2\cdot 2^{n-2}+r']}=\theta^{r'}(11); \\
C^{[n]}_{[3\cdot 2^{n-2}+r']}=\theta^{r'}(10).
\end{array}
\end{equation}
Thanks to the prop. (i) of Lemma \ref{C-period}, the eqs. (\ref{Equations-C}), (\ref{Equations-Cn}) directly provide the value of the component $C^{[j]}_{[i]}$, for every $i\in [0,2^{n-1}]$.
However, from an algorithmic point of view, in order to compute the whole matrix $C$, all the words $C^{[j]}_{[i]}$ should be memorized, for all the integer pairs $i\in [0,2^{j+1}-1]$, $j\in J$.
In other words we are still a long way from our goal. 
In view of that, in  what follows  we shall deepen the structure of Eqs. (\ref{Equations-C}), (\ref{Equations-Cn}). 
First of all we note that, given some index $i\in [1,2^{n-1}]$, exactly one of the two following conds. occurs:

\medbreak
\noindent
{\underline {\it Condition $i\ne 0\bmod 2^{j+1}$}}\\
With the preceding notation, this condition corresponds to $\mu(i,j)\ge 1$,  thus  $r'\ge 1$. 
According to  Eqs. (\ref{Equations-C}), (\ref{Equations-Cn}), for every $j\in J\setminus\{n_0\}$ and for each integer pair $q'\in [0,7]$, $r'\in [0,2^{j-2}-1]$, 
a unique word $x\in A^2$ exists st. $C^{[j]}_{[q'^{j-2}+r']}=\theta^{r'}(x)$ and $C^{[j]}_{[q'^{j-2}+(r'-1)]}=\theta^{r'-1}(x)$, therefore the following equation holds:
\begin{eqnarray}
\label{Induction-1}
C^{[j]}_{[q'2^{j-2}+r']}=\theta\left(C^{[j]}_{[q'2^{j-2}+(r'-1)]}\right)~~(1\le r'\le 2^{j-2}-1).
\end{eqnarray}

\medbreak
\noindent
{\underline{\it Condition $i=0\bmod 2^{j+1}$}}\\
Actually, with this condition, which  is equivalent to $r'=0$, i.e. $\mu(i,j)=0$, we are not able to directly establish a formula similar to Eq. (\ref{Induction-1}).
For instance  in Eqs. (\ref{Equations-C}), on the first hand, by taking $q'=3$, we have  
$C^{[j]}_{[3\cdot 2^{j-2}-1]}=C^{[j]}_{[2\cdot 2^{j-2}+(2^{j-2}-1)]}=\theta^{2^{j-2}-1}(11)=00$ and $C^{[j]}_{[q'3\cdot 2^{j-2}]}=10$.
On the other hand, by taking $q'=6$, although still we have  
$C^{[j]}_{[6\cdot 2^{j-2}-1]}=C^{[j]}_{[5\cdot 2^{j-2}+(2^{j-2}-1)]}=\theta^{5^{j-2}-1}(11)=00$,  we have in fact $C^{[j]}_{[q'6^{j-2}]}=01$.
In other words, under current conditions, the formulas (\ref{Equations-C}) and (\ref{Equations-Cn})
 cannot provide sufficient information to express $C^{[j]}_{[q'2^{j-2}]}$ directly, starting with $C^{[j]}_{[q'2^{j-2}-1]}$.
\subsection{Some breakthrough  thanks to a new matrix}
In order to gather the missing information, we introduce a second matrix, namely $Q$. 
With the  above  notation, given an index pair $i\in [0,2^n-1]$, $j\in J$, we set $Q^{[j]}_{[i]}=q'$ that is, 
$\mu(i,j)=Q^{[j]}_{[i]}2^{j-2}+r'$.
The following equations 
come from  Eqs. (\ref{Equations-C}): 

\begin{equation}
\label{Equations-q-C}
\begin{array} {c} 
~~~~~~~~~~~~~~~~~  \left(Q^{[j]}_{[r']},C^{[j]}_{[r']}\right)=\left(0,\theta^{r'}(00)\right)\\
~~~\left(Q^{[j]}_{[2^{j-2}+r']},C^{[j]}_{[2^{j-2}+r']}\right)=\left(1,\theta^{r'}(01)\right) \\
\left(Q^{[j]}_{[2\cdot 2^{j-2}+r']},C^{[j]}_{[2\cdot 2^{j-1}+r']}\right)=\left(2,\theta^{r'}(11)\right)\\
~\left(Q^{[j]}_{[3\cdot 2^{j-2}+r']},C^{[j]}_{[3\cdot^{j-2}+r']}\right)=\left(3,\theta^{r'}(10)\right)\\
%
~~~~~~~~ \left(Q^{[j]}_{[2^{j}+r']},C^{[j]}_{[2^{j}+r']}\right)=\left(4,\theta^{r'}(10)\right)\\
\left(Q^{[j]}_{[5\cdot 2^{j-2}+r']},C^{[j]}_{[5\cdot 2^{j-2}+r']}\right)=\left(5,\theta^{r'}(11)\right)\\
\left(Q^{[j]}_{[6\cdot 2^{j-2}+r']},C^{[j]}_{[6\cdot 2^{j-2}+r']}\right)=\left(6,\theta^{r'}(01)\right)\\
\left(Q^{[j]}_{[7\cdot 2^{j-2}+r']},C^{[j]}_{[7\cdot 2^{j-2}+r']}\right)=\left(7,\theta^{r'}(00)\right).
\end{array}
\end{equation}

\medskip
\noindent
In addition, for $j=n$, the following equations come from Eqs.  (\ref{Equations-Cn}):

\begin{equation}
\label{Equations-q-Cn}
\begin{array} {c} 
~~~~~~~~~~~~~~~~   \left(Q^{[n]}_{[r']},C^{[n]}_{[r']}\right)=\left(0,\theta^{r'}(00)\right) \\
~~~\left(Q^{[n]}_{[2^{n-2}+r']},C^{[n]}_{[2^{n-2}+r']}\right)=\left(1,\theta^{r'}(01)\right) \\
~~~ \left(Q^{[n]}_{[2^{n-1}+r']},C^{[n]}_{[2^{n-1}+r']}\right)=\left(2,\theta^{r'}(10)\right) \\
\left(Q^{[n]}_{[3\cdot 2^{n-2}+r']},C^{[n]}_{[3\cdot 2^{n-2}+r']}\right)=\left(3,\theta^{r'}(11)\right).
\end{array}
\end{equation}

\medskip
\noindent
Regarding periodicity of the sequences, the following prop. comes from  Lemma \ref{C-period}:
\begin{lem}
\label{q-C-period}
For every $j\in J\setminus\{n_0\}$ the sequence $\left(Q^{[j]}_{[i]},C^{[j]}_{[i]}\right)_{0\le i\le 2^n-1}$ is $2^{j+1}$-periodic.
\end{lem}
\begin{proof}
Let $i_1,i_2\in [0,2^n-1]$ st. $i_1-i_2=0\bmod 2^{j+1}$.
By definition the integers $\mu(i_1,j),\mu(i_2,j)\in [0,2^{j+1}-1]$ are equal, whence $Q^{[j]}_{[i_1]}$, $Q^{[j]}_{[i_2]}$, their corresponding euclidian quotients by $2^{j-2}$ are equal.  
In addition, according to prop. (i) of Lemma \ref{C-period}, the sequence $C^{[j]}_{[0\cdot\cdot2^n-1]}$ is $2^{j+1}$-periodic, therefore we have $C^{[j]}_{[i_1]}=C^{[j]}_{[i_2]}$.
\end{proof}
\bigbreak
\noindent
In view of  Eqs. (\ref{Equations-q-C}), (\ref{Equations-q-Cn}), and Lemma \ref{q-C-period}, we introduce  the following $8$-element cycle: 
$$\pi=\left(\left(0,00\right), \left(1,01\right),\left(2,11\right),\left(3,10\right),\left(4,10\right),\left(5,11\right),\left(6,01\right),\left(7,00\right)\right)$$
that is,   $\pi(\left(0,00\right)=\left(1,01\right), \pi\left(1,01\right)=\left(2,11\right), \cdots, \pi\left(6,01\right)=\left(7,00\right), \pi\left(7,00\right)=(\left(0,00\right) $.

\bigbreak\noindent
Recall that we have $\mu(i,j)=Q^{[j]}_{[i]}2^{j-2}+r'$.
\begin{lem}\label{C}
With the preceding notation, for every integer pair $i\in [1,2^{n}-1]$, $j\in J\setminus\{n_0\}$, the following  equation holds:
\begin{eqnarray}
\left(Q^{[j]}_{[i]}, C^{[j]}_{[i]}\right)=
\begin{cases}
\left(Q^{[j]}_{[i-1]},\theta\left(C^{[j]}_{[i-1]}\right)\right)~~{\rm if}~~r'\ne 0,\\
\pi\left(Q^{[j]}_{[i-2^{j-2}]}, C^{[j]}_{[i-2^{j-2}]}\right)~~{\rm otherwise.}
\end{cases}
\nonumber
\end{eqnarray}
\end{lem}
\begin{proof}
Let $j\in J\setminus\{n_0\}$.
According to the value of $r'\in [0,2^{j}-1]$ exactly one of the two following conditions occurs:

\medbreak
\noindent
(i) {\underline{\it Condition  $r'\ne 0$}}\\
According to  Eqs.  (\ref{Equations-q-C}), (\ref{Equations-q-Cn}), 
we obtain  $\left(Q^{[j]}_{[\mu(i,j)]},C^{[j]}_{[\mu(i,j)]}\right)=\left(Q^{[j]}_{[\mu(i,j)-1]},\theta\left(C^{[j]}_{[\mu(i,j)-1]}\right)\right)=\left(Q^{[j]}_{[\mu(i-1,j)]}\theta\left(C^{[j]}_{[\mu(i-1,j)]}\right)\right)$.
According to Lemma \ref{q-C-period}, we  have $C^{[j]}_{[i]}=C^{[j]}_{[\mu(i,j)]}$ and $C^{[j]}_{[i-1]}=C^{[j]}_{[\mu(i-1,j)]}$:
this implies  $\left(Q^{[j]}_{[i]}, C^{[j]}_{[i]}\right)=\left(Q^{[j]}_{[i-1]},\theta\left(C^{[j]}_{[i-1]}\right)\right)$.

\bigbreak\noindent
(ii) {\underline {\it Condition  $r'=0$}}\\
According to the value of the index $i$, exactly one of the three following cases occurs:

\medbreak (ii.i) {\underline {\it The case where we have  $i\in [1,2^{j+1}-1]$}}

According to Eqs.  (\ref{Equations-q-C}) each of the following identities hold (in the case where we have $j=n$, only the first four hold):
\begin{equation}
\label{Equations-q-C-r'=0}
\begin{array} {c} 
~~~~~\left(Q^{[j]}_{[2^{j-2}]},C^{[j]}_{[2^{j-2}]}\right)=\left(1,01\right) \\
~\left(Q^{[j]}_{[2\cdot 2^{j-2}]},C^{[j]}_{[2\cdot 2^{j-1}]}\right)=\left(2,11\right)\\ 
~~\left(Q^{[j]}_{[3\cdot 2^{j-2}]},C^{[j]}_{[3\cdot^{j-2}]}\right)=\left(3,10\right)\\
~~~~~~~~~~\left(Q^{[j]}_{[2^{j}]},C^{[j]}_{[2^{j}]}\right)=\left(4,10\right)\\
~\left(Q^{[j]}_{[5\cdot 2^{j-2}]},C^{[j]}_{[5\cdot 2^{j-2}]}\right)=\left(5,11\right)\\
~\left(Q^{[j]}_{[6\cdot 2^{j-2}]},C^{[j]}_{[6\cdot 2^{j-2}]}\right)=\left(6,01\right)\\
~\left(Q^{[j]}_{[7\cdot 2^{j-2}]},C^{[j]}_{[7\cdot 2^{j-2}]}\right)=\left(7,00\right).
\end{array}
\end{equation}
According to Eqs. (\ref{Equations-q-C-r'=0}), for every $q'\in [1,7]$ we have:
 $$\left(Q^{[j]}_{[q'\cdot 2^{j-2}]},C^{[j]}_{[q'\cdot 2^{j-2}]}\right)=\pi\left(Q^{[j]}_{[(q'-1)\cdot 2^{j-2}]},C^{[j]}_{[(q'-1)\cdot 2^{j-2}]}\right).$$
On the other hand, by definition $i\in [1,2^{j+1}-1]$ implies $i=\mu(i,j)=Q^{[j]}_{[i]}\cdot 2^{j-2}+r'=q'\cdot 2^{j-2}+r'=q'\cdot 2^{j-2}$,
thus $i-2^{j-2}=(q'-1)2^{j-2}$.
We obtain:
$$\left(Q^{[j]}_{[i]},C^{[j]}_{[i]}\right)=\left(Q^{[j]}_{[q'\cdot 2^{j-2}]},C^{[j]}_{[q'\cdot 2^{j-2}]}\right)=
\pi\left(Q^{[j]}_{[i-2^{j-2}]},C^{[j]}_{[i-2^{j-2}]}\right).$$

(ii.ii) {\underline {\it  The case where $i=0\bmod 2^{j+1}$}}

On the one hand, according to Lemma \ref{q-C-period}, we have 
 $\left(Q^{[j]}_{[i]},C^{[j]}_{[i]}\right)=
\left(Q^{[j]}_{[0]},C^{[j]}_{[0]}\right)=(0,00)$.
On the other hand, we have  $i-2^{j-2}=-2^{j-2}\bmod 2^{j+1}=2^{j+1}-2^{j-2}\bmod 2^{j+1}$,
thus $i-2^{j-2}=8\cdot 2^{j-2}-2^{j-2}\bmod 2^{j+1}=7\bmod 2^{j+1}$.
According to Lemma \ref{q-C-period} and Eqs. (\ref{Equations-q-C-r'=0}), we obtain $\left(Q^{[j]}_{[i-2^{j-2}]},C^{[j]}_{[i-2^{j-2}]}\right)=\left(Q^{[j]}_{[7\cdot 2^{j-2}]},C^{[j]}_{[7\cdot 2^{j-2}]}\right)=(7,00)$: 
once more we have  $\left(Q^{[j]}_{[i]},C^{[j]}_{[i]}\right)=\pi\left(Q^{[j]}_{[i-2^{j-2}]},C^{[j]}_{[i-2^{j-2}]}\right)$.

\medbreak
(ii.iii) {\underline {\it The case where $i\in [2^{j+1}+1,2^{n}-1]$, with $i\ne 0\bmod 2^{j+1}$}}

By definition we have $\mu(i,j)\in [1,2^{j+1}-1]$.
On the one hand, by substituting  $\mu(i,j)$ to $i$ in the preceding case (ii.i), we obtain:
$\left(Q^{[j]}_{[\mu(i,j)]},C^{[j]}_{[\mu(i,j)]}\right)=\left(Q^{[j]}_{[\mu(i,j)-2^{j-2}]},C^{[j]}_{[\mu(i,j)-2^{j-2}]}\right)$.
On the other hand, by the definition of $\mu$ we have $i=q\cdot 2^{j+1}+\mu(i,j)$, thus $i-2^{j-2}=q\cdot 2^{j+1}+\left(\mu(i,j)-2^{j-2}\right)$.
Since we have  $r'=0$ and $\mu(i,j)\ge 1$ a positive integer $q'\ge1$ exists st. $\mu(i,j)=q'\cdot2^{j-2}$.
It follows from  $\mu(i,j)-2^{j-2}=(q'-1)2^{j-2} \ge 0$ and $\mu(i,j)-2^{j-2}\le 2^{j+1}-1$ that $\mu(i,j)-2^{j-2}\in [0,2^{j+1}-1]$.
By the definition of $\mu$ we obtain $\mu(i-2^{j-2},j)= \mu(i,j)-2^{j-2}$, thus $\left(Q^{[j]}_{[\mu(i,j)-2^{j-2}]},C^{[j]}_{[\mu(i,j)-2^{j-2}]}\right)=\left(Q^{[j]}_{[\mu(i-2^{j-2},j)]},C^{[j]}_{[\mu(i-2^{j-2},j)]}\right)$.
According to Lemma \ref{q-C-period}, we  obtain:
 $$\left(Q^{[j]}_{[i]},C^{[j]}_{[i]}\right)= \left(Q^{[j]}_{[\mu(i,j)]},C^{[j]}_{[\mu(i,j)]}\right)=
\left(Q^{[j]}_{[\mu(i-2^{j-2},j)]},C^{[j]}_{[\mu(i-2^{j-2},j)]}\right)= \pi\left(Q^{[j]}_{[i-2^{j-2}]},C^{[j]}_{[i-2^{j-2}]}\right).$$
This completes the proof of Lemma \ref{C}.
\end{proof}
\bigbreak\noindent

Lemma \ref{C} shows that the condition $r'=0$ plays a prominent part in the computation of $\left(Q^{[j]}_{[i]}, C^{[j]}_{[i]}\right)$.
The following result provides some precision:
\begin{lem}
\label{mu-ou-i}
With the preceding notation, the three following conditions are equivalent:

{\rm (i)}  $r'=0$;

{\rm (ii)} $\mu(i,j)=0\bmod 2^{j-2}$;

{\rm (iii)} $i=0\bmod 2^{j-2}$.
\end{lem}
\begin{proof}
It follows from $\mu(i,j)=Q^{[j]}_{[i]}2^{j-2}+r'$, with $r'\in [0,2^{j-2}-1]$, that  the two conditions $\mu(i,j)=0\bmod 2^{j-2}$ and $r'=0$ are equivalent.
%
According to the definition of $\mu$,
some integer $m\in{\mathbb N}$ exists st. $i=\mu(i,j)+m 2^{j+1}$.
As a consequence, if we have $i=0\bmod 2^{j-2}$, some integer $m'\in{\mathbb N}$ exists st. $\mu(i,j)+m 2^{j+1}=m'2^{j-2}$, thus $\mu(i,j)=(m'-8m)2^{j-2}$ that is,  $\mu(i,j)=0\bmod 2^{j-2}$.
Conversely if $m'\in{\mathbb N}$ exists st. $\mu(i,j)=m' 2^{j-2}$, we obtain $i=m'2^{j-2}+m2^{j+1}$ that is, $i=(m'+8m)2^{j-2}$, thus  $i=0\bmod 2^{j-2}$.
\end{proof}
\subsection{The loopless algorithm}
Nevertheless, we have not fully achieved our objective: indeed, in  the statement of Lemma \ref{C}, in order to compute  $\left(Q^{[j]}_{[i]}, C^{[j]}_{[i]}\right)$, 
the condition $i=0\bmod 2^{j-2}$ imposes to memorize the component $\left(Q^{[j]}_{[i-2^{j-2}]}, C^{[j]}_{[i-2^{j-2}]}\right)$.
Our goal is to prove that such a computation can be actually done by only referring to the pair $\left(Q^{[j]}_{[i-1]}, C^{[j]}_{[i-1]}\right)$.
In order to do so  we need to introduce some additional concept: 
we denote by $\phi$ be the partial mapping onto $[0,7]\times A^2$ defined  by $\phi\left(q',\theta^{-1}(c)\right)=\pi\left(q',c\right)$, for each pair $(q',c)$ in the cycle $\pi$.
By definition  $\phi$ takes the following values:

\begin{equation}
\label{phi}
\arraycolsep=5pt
\begin{array} {|l|*{9}{r@{}l|}}\hline
(q',c) && \left(0,11\right) & & \left(1 ,10\right) & & \left(2,00\right) & & \left(3,01\right) & & \left(4,01\right) & & \left(5,00\right) & &\left(6,10\right) & &\left(7,11\right) \\ \hline
\phi(q',c) &&  \left(1,01\right) & & \left(2,11\right) & & \left(3,10 \right) & &\left(4,10\right) & &\left(5,11\right) & & \left(6,01\right) & & \left(7,00\right)& & \left(0,00\right) \\ \hline 
\end{array}
\end{equation}

\bigbreak
\noindent
The following property is the basis to an iterative algorithm to compute the whole sequence $\gamma^{n,k}$:
\begin{proposit}
\label{CC}
With the preceding notation, for every $i\in [1,2^{n}-1]$, $j\in J\setminus\{n_0\}$, the following  identity holds:
\begin{eqnarray}
\left(Q^{[j]}_{[i]}, C^{[j]}_{[i]}\right)=
\begin{cases}
\left(Q^{[j]}_{[i-1]},\theta\left(C^{[j]}_{[i-1]}\right)\right)~~{\rm if}~~i\ne 0\bmod 2^{j-2}\\
\phi\left(Q^{[j]}_{[i-1]}, C^{[j]}_{[i-1]}\right)~~ {\rm otherwise.}
\end{cases}
\nonumber
\end{eqnarray}
\end{proposit}

\begin{proof}
According to  Lemmas \ref{C}, \ref{mu-ou-i} we restrain to the case where we have $i=0\bmod 2^{j-2}$ that is, with the preceding notation,
$r'=0$ and  $\mu(i,j)=q'\cdot  2^{j-2}$. 
As in proof of Lemma \ref{q-C-period} exactly one of the three following conditions holds:
\medbreak
\noindent
(i) {\underline {\it The case where we have $i\in [1, 2^{j+1}-1]$}}\\
With this condition we have $ i=\mu(i,j)\ge 1$. 
Firstly, we assume $j<n$. According to Eqs. (\ref{Equations-q-C-r'=0}), and by the definition of $\phi$
each of the following identities holds: 
\begin{equation}
\begin{array}{c}
\label{Equations-q-C1}%
~~~~\left(Q^{[j]}_{[2^{j-2}]},C^{[j]}_{[2^{j-2}]}\right)=\left(1,01\right)=\pi(0,00)=\phi\left(0,11\right)
\\
 ~\left(Q^{[j]}_{[2\cdot 2^{j-2}]},C^{[j]}_{[2\cdot2^{j-1}]}\right)=\left(2,11\right)=\pi(1,01)=\phi\left(1,10\right)
\\
~\left(Q^{[j]}_{[3\cdot 2^{j-2}}],C^{[j]}_{[3\cdot 2^{j-2}]}\right)=\left(3,10\right)=\pi(1,11)=\phi\left(2,00\right)\\
~\left(Q^{[j]}_{[4\cdot 2^{j-2}]},C^{[j]}_{[4\cdot 2^{j-2}]}\right)=\left(4,10\right)=\pi(3,10)=\phi\left(3,01\right)\\
~\left(Q^{[j]}_{[5\cdot 2^{j-2}]},C^{[j]}_{[5\cdot 2^{j-2}]}\right)=\left(5,11\right)=\pi(4,10)=\phi\left(4,01\right)\\
~\left(Q^{[j]}_{[6\cdot 2^{j-2}]},C^{[j]}_{[6\cdot 2^{j-2}]}\right)=\left(6,01\right)=\pi(5,11)=\phi\left(5,00\right)\\
~\left(Q^{[j]}_{[7\cdot 2^{j-2}]},C^{[j]}_{[7\cdot 2^{j-2}]}\right)=\left(7,00\right)=\pi(6,01)=\phi\left(6,10\right).
\end{array}
\end{equation}
In the case where we have $j=n$, only the first four  equations hold.
It is straightforward to verify that , for each $q'\in [1,7]$, we have, 
$\left(Q^{[j]}_{[q'2^{j-2}]},C^{[j]}_{[q'2^{j-2}]}\right)=\phi\left(Q^{[j]}_{[q'2^{j-2}-1)]},C^{[j]}_{[q'2^{j-2}-1)]}\right)$ that is,
$\left(Q^{[j]}_{i},C^{[j]}_{[i]}\right)=\phi\left(Q^{[j]}_{[i-1]}, (Q^{[j]}_{[i-1]} \right)$.

\medbreak\noindent
(ii) {\underline {\it The case where  $i=0\bmod 2^{j+1}$}}\\ 
 Let $q$ be the unique positive integer st. $i=q\cdot 2^{j+1}$.
  On the one hand, according to Lemma \ref{q-C-period}, we have 
 $\left(Q^{[j]}_{[i]},C^{[j]}_{[i]}\right)=
\left(Q^{[j]}_{[0]},C^{[j]}_{[0]}\right)=(0,00)=\phi(7,11)$.
On the other hand, we have  $i-1=q\cdot 2^{j+1}-1=(q-1)2^{j+1}+(2^{j+1}-1)$.  
It follows from $j\ge 1$ that $2^{j+1}-1\in [0,  2^{j+1}-1]$, whence we have $\mu(i-1,j)=2^{j+1}-1$.
This implies $\left(Q^{[j]}_{[i-1]},C^{[j]}_{[i-1]}\right)=
\left(Q^{[j]}_{[2^{j+1}-1]},C^{[j]}_{[2^{j+1}-1]}\right)$:
we are in the condition of Eqs. (\ref{Equations-q-C}) with $q'=7$ and $r'=2^{j-2}-1$.
We obtain $\left(Q^{[j]}_{[2^{j+1}-1]},C^{[j]}_{[2^{j+1}-1]}\right)=\left(Q^{[j]}_{7\cdot 2^{j-2}+(2^{j-2}-1)]},C^{[j]}_{[7\cdot 2^{j-2}+(2^{j-2}-1)]}\right)=\left(7,\theta^{2^{j-2}-1}(00)\right)=(7,11)$.
As a consequence, once more we have $\left(Q^{[j]}_{[i]},C^{[j]}_{[i]}\right)=\phi\left(Q^{[j]}_{[i-1]},C^{[j]}_{[i-1]}\right)$.

\medbreak\noindent
(iii) {\underline {\it The case where  $i\in [2^{j+1}+1,2^{n}-1]$ and $i\ne 0\bmod 2^{j+1}$}}\\
On the one hand, with this condition we have $\mu(i,j)\in [1,2^{j+1}-1]$. By substituting $\mu(i,j)$ to $i$ in  the preceding case (i) we obtain:

 $\left(Q^{[j]}_{[\mu(i,j)]},C^{[j]}_{[\mu(i,j)]}\right)=\phi\left(Q^{[j]}_{[\mu(i,j)-1]},C^{[j]}_{[\mu(i,j)-1]}\right)$.

On the other hand,   by the definition of $\mu$ we have $i=q\cdot 2^{j+1}+\mu(i,j)$:
 this implies $i-1=q\cdot2^{j+1}+\left(\mu(i,j)-1\right)$.  
It follows from $i\ne 0\bmod 2^{j+1}$ that $\mu(i,j)-1\in [0,2^{j+1}-1]$, therefore we have $\mu(i,j)-1=\mu(i-1,j)$.
As a consequence, we obtain $\left(Q^{[j]}_{[\mu(i,j)]},C^{[j]}_{[\mu(i,j)]}\right)=\phi\left(Q^{[j]}_{[\mu(i-1,j)]},C^{[j]}_{[\mu(i-1,j)]}\right)$. 
According to Lemma \ref{q-C-period}, this implies $\left(Q^{[j]}_{[i]},C^{[j]}_{[i]}\right)=\phi\left(Q^{[j]}_{[i-1]},C^{[j]}_{[i-1]}\right)$: this completes the proof.
\end{proof}
\bigbreak\noindent
According to the result of Proposition \ref{CC} we obtain an iteration-based method for computing the sequence $\gamma^{n,k}$ (see Algorithm 2). 
Recall that we set $C^{[n_0]}_{[0]}=\gamma^{n_0,1}$. 
From the point of view of implementation,  in the spirit of Algorithm 1, 
for every $i\in [0,2^n-1]$ the two following objects:

-- The component $C^{[n_0]}_{[i]}$ 

-- the row  $\left(\left(Q^{[n]}_{[i]},C^{[n]}_{[i]}\right),\left(Q^{[n-2]}_{[i]},C^{[n-2]}_{[i]}\right),\cdots, \left(Q^{[n_0+2]}_{[i]},C^{[n_0+2]}_{[i]}\right)\right)$,

are memorized in the corresponding  generic components:

-- ${\cal C}^{[n_0]}$

-- $\left(\left({\cal Q}^{[n]},{\cal C}^{[n]}\right),\left({\cal Q}^{[n-2]},{\cal C}^{[n-2]}\right),\cdots, \left({\cal Q}^{[n_0+2]},{\cal C}^{[n_0+2]}\right)\right)$.

Each time the counter $i$ is incremented these generic components are updated.
According to Eqs. \ref{Eq-C}, the column $C^{[n_0]}$ takes the following expression:
\begin{equation}
\label{C0}
C^{[n_0]}_{[0\cdot\cdot2^{n_0+1}-1]}=\left(\gamma^{n_0,1},\rho^{n_0,1}\right),~~
C^{[n_0]}=\left(\gamma^{n_0,1},\rho^{n_0,1}\right)^{2^{n-n_0-1}}.
\end{equation}
From this point of view, we start the computation by setting:
 ${\cal C}^{n_0}=C^{[n_0]}_{[0]}=\gamma^{n_0,1}_{[0]}$.
\medbreak\noindent
{\it Some comments about Algorithm 2}\\
The variable $b$  takes values in the  two-symbol set $\{\gamma,\rho\}$: 
its role is to determine which of  Eqs. (\ref{Equations-q-C}), (\ref{Equations-q-Cn}) should be applied  in order to  compute the value of the pair $\left(Q^{[n_0}_{[i]},C^{[n_0]}_{[i]}\right)$ (lines 7--11).
According to the prop. (ii) of Lemma \ref{C-period},  the variable $b$, which is initialized to $\gamma$, is actualized each time  the integer $\mu_0=\mu(i,2^{n_0+2})$ meets some element of $\{ 2^{n_0},2^{n_0+1}\}$ (lines 12--17).
  The result of Proposition \ref{CC}, for its part, is applied at lines 19--26.
\medbreak
\noindent
{\it Questions related to complexity}\\
The study is similar to the one of Sect. \ref{A>3bis}, and it leads to similar conclusions. 
Beforehand we note that, in any case, the alphabet $A$ and the permutation $\theta$ should be computed in a preprocessing phase.

\smallbreak
-- Regarding the generic row $\left(\left({\cal Q}^{[n]},{\cal C}^{[n]}\right),\cdots, \left({\cal Q}^{[n_0+1]},{\cal C}^{[n_0+1]}\right)\right)$, 
there is a positive integer, say $\ell$ (the maximum cost of each operation over every component),
st.  updating the sequence requires at most $\ell(n-n_0-1)= \ell k$ insertions. 
Consequently, when the counter $i$ reaches the value $i_{max}+1$, the total amount of operations is at most $2^n\ell k$.
 
\smallbreak
-- In order to compute the component  ${\cal C}^{[n_0]}$, as for Algorithm 1 there are  two possible approaches: 

\smallbreak
(a) Firstly, for each value of $i\in [1,2^{n_0}-1]$ we apply the instruction (\ref{choix-j-g}) from Algorithm (b) (see
the preliminaries). Computing each of the finite sequences  $\gamma^{n_0,1}$,  $\rho^{n_0,1}$,
classically requires an amount of   $2^{n_0}+2^{n_0-1}+\cdots+2\le 2^{n_0}$
one-character   substitutions, therefore for computing
the whole column $C^{[n_0]}$, the total cost of the preceding operations
is bounded by $2^{n-n_0}\cdot 2\cdot 2^{n_0}=2^{n+1}$.
Consequently the total amount of operations is bounded by $2^{n+1}\ell k+2^{n+1}=2^{n+1}(\ell k+1)$.
This leads to a computation in  amortized time of $2^{-n}2^{n+1}(\ell k+1)=O(1)$, with  space linear in $n$.

\smallbreak
(b)  The second approach consists in implementing in a preprocessing phase the sequences $\gamma^{n_0,1}$, $\rho^{n_0,1}$ and  the mappings $\pi$, $\phi$ :
such an implementation requires  space $O(n_02^{n_0})$ and, as indicated above,  a total amount of $O(2^{n_0})$ substitutions, with space  $O(n_02^{n_0})$.
After that, in the processing phase, updating ${\cal C}^{n_0}$ will be performed by
a constant number of requests to $\gamma^{n_0,1}$, $\gamma^{n_0,1}$, $\pi$ and $\phi$, say $\ell_1$ (see lines 19-23).
Consequently, in the processing phase updating the matrices ${\cal C}^{[n..n_0+1]}$ , ${\cal Q}^{[n..n_0+1]}$ requires  at most  $\ell +\ell_1k$ operations:
with this second strategy  of implementation  Algorithm 2 is loopless and  requires space linear in $n+n_02^{n_0}$.
\begin{figure}[H]
\begin{center}
\includegraphics[width=15cm,height=12.5cm]{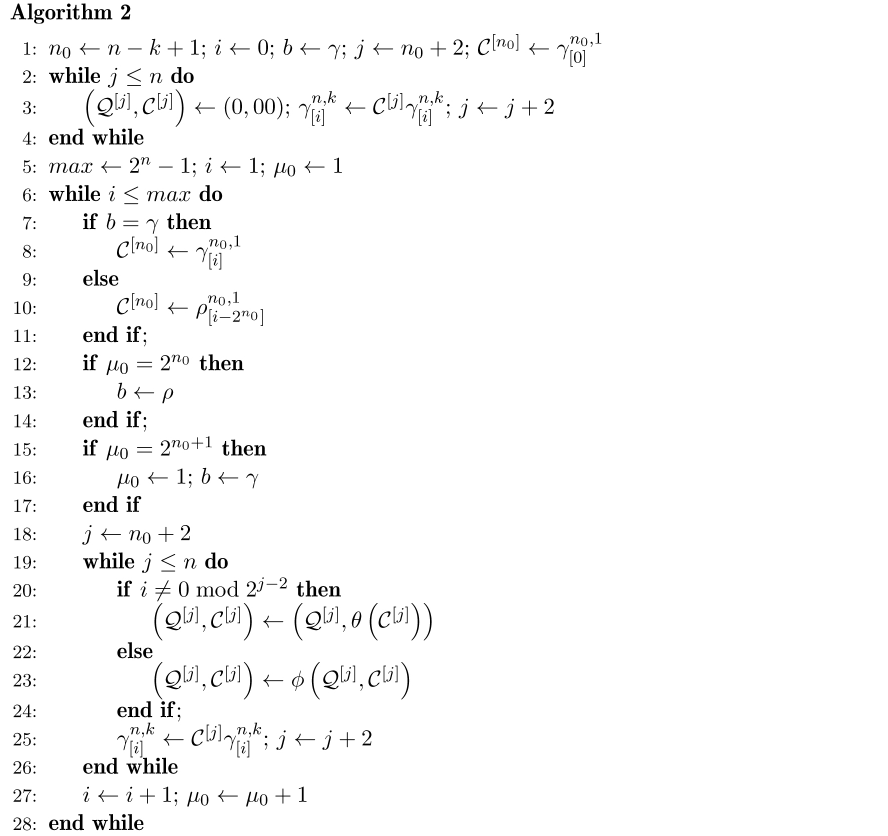}
\end{center}
\end{figure}
\medbreak
\begin{figure}[H]
\begin{center}
\includegraphics[width=12cm,height=7cm]{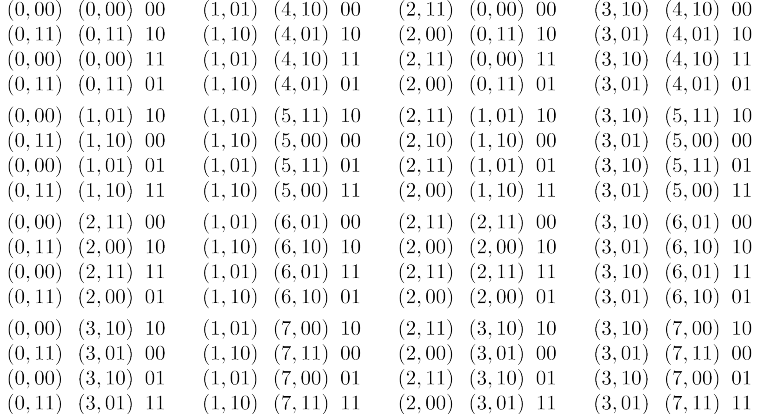}
\end{center}
\end{figure}

\section{The case where we have $|A|=2$ and  $k$ even}
\label{Cons}
Let $k$ be a positive even integer.
Beforehand,  we remind some classical algebraic interpretation of the substitution $\sigma_k$  in the framework of the binary alphabet $A=\{0,1\}$.
Denote by $\oplus$ the addition in the group ${\mathbb Z}/2{\mathbb Z}$ with identity $0$. 
Given a positive integer $n$, 
and  $w,w'\in A^n$, define $w\oplus w'$ as the unique word of $A^n$ st.
$(w\oplus w')_i=w_i\oplus w'_i$, for each $i\in [1,n]$. With this notation the sets $A^n$ and $({\mathbb Z}/2{\mathbb Z})^n$ are in one-to-one correspondence. Moreover we have
$w'\in\sigma_k(w)$ iff. some word $u\in A^n$ exists st. $|u|_1=k$ and $w=w' {\oplus} u$:  since $k$ is even,  we obtain $|w|_1=|w'|_1\bmod 2$.
Consequently, given a  $\sigma_k$-Gray cycle $\left(\alpha_{[i]}\right)_{0\le i\le m}$, 
for each $i\in [0,m]$ the equation $\left|\alpha_{[i]}\right|_1=\left|\alpha_{[0]}\right|_1\bmod 2$ holds.
As a corollary, setting ${\rm Even}_1^n=\{w\in A^*:|w|_1=0 \bmod 2\}$ and ${\rm Odd}_1^n=\{w\in A^*:|w|_1=1\bmod 2\}$, we obtain the following property:
\begin{lem}
\label{pary}
With the condition of Section \ref{Cons},  given a  $\sigma_k$-Gray cycle $\alpha$ over $X$, either we have $X\subseteq {\rm Even}_1^n$, or we have $X\subseteq {\rm Odd}_1^n$.
\end{lem}
Given an even integer $n$, we define the sequences $\gamma^{n,k}$ and $\underline{\gamma}^{n,k}$ as indicated in the following:
\begin{eqnarray}
\label{Construct}
(\forall i\in [0, 2^{n-1}-1])~~\gamma^{n,k}_{[i]}=\theta^i(0)\gamma^{n-1,k-1}_{[i]}~~{\rm and}
~~\underline{\gamma}^{n,k}_{[i]}=\theta^i(1)\gamma^{n-1,k-1}_{[i]}.
\end{eqnarray}
 According to Proposition \ref{Gamma-n-k-odd}, since $k-1$ is an odd integer the sequence $\gamma^{n-1,k-1}$
is a  $\sigma_{k-1}$-Gray cycle over $A^{n-1}$.

For instance, we have  $\gamma^{6,4}_{[0]}=0 00 000$, $\underline\gamma^{6,4}_{[0]}=1 00 000$, $\gamma^{6,4}_{[1]}=111100$, and $\underline\gamma^{6,4}_{[1]}=011100$.
\begin{proposit}
\label{Gamma-n-k-even}
The sequence  $\gamma^{n,k}$ (resp., $\underline{\gamma}^{n,k}$) is a $\sigma_k$-Gray cycle over  ${\rm Even}_1^{n}$ (resp., ${\rm Odd}^{n}_1$).
\end{proposit}
\begin{proof}
(i) According to Eqs. (\ref{Construct}), 
 since $\gamma^{n-1,k-1}$ satisfies Cond. \ref{ivc}, by construction both the sequences $\gamma^{n,k}$ and $\underline{\gamma}^{n,k}$ also satisfy \ref{ivc}.

\medbreak\noindent
(ii) By Lemma \ref{pary}, we have $\bigcup_{0\le i\le 2^{n}-1}\left\{\gamma^{n,k}\right\}\subseteq{\rm Even}_1^{n}$
and $\bigcup_{0\le i\le 2^{n}-1}\left\{\underline{\gamma}^{n,k}\right\}\subseteq {\rm Odd}_1^{n}$.
In addition, according to Eqs. (\ref{Construct}), we have  $\left|\gamma^{n,k}\right|=\left|\underline{\gamma}^{n,k}\right|=\left|\gamma^{n-1,k-1}\right|=2^{n-1}=\left| {\rm Even}_1^{n}\right|$. This implies 
 $\bigcup_{0\le i\le 2^{n}-1}\left\{\gamma^{n,k}\right\}= {\rm Even}_1^{n}$
and $\bigcup_{0\le i\le 2^{n}-1}\left\{\underline{\gamma}^{n,k}\right\}= {\rm Odd}_1^{n}$ that is, both the sequences $\gamma^{n,k}$ and $\underline\gamma^{n,k}$ satisfy Cond. \ref{iva}.

\medbreak\noindent
(iii) Let $i\in [1, 2^{n-1}-1]$. Since $\gamma^{n-1,k-1}$ satisfies \ref{ivb}, we have $\gamma^{n-1,k-1}_{[i]}\in\sigma_{k-1}\left(\gamma^{n-1,k-1}_{[i-1]}\right)$.
According to Eqs. (\ref{Construct}), the initial characters of $\gamma^{n,k}_{[i]}$
and  $\gamma^{n,k}_{[i-1]}$ (resp.,  $\underline\gamma^{n,k}_{[i]}$
and  $\underline\gamma^{n,k}_{[i-1]}$) are different, hence we have $\gamma^{n,k}_{[i]}\in\sigma_k\left(\gamma^{n,k}_{[i-1]}\right)$ and $\underline\gamma^{n,k}_{[i]}\in\sigma_k\left(\underline\gamma^{n,k}_{[i-1]}\right)$.
In addition, once more according to Eqs. (\ref{Construct}) it follows from $\gamma^{n-1,k-1}_{[0]}\in\sigma_{k-1}\left(\gamma^{n-1,k-1}_{[2^{n-1}-1]}\right)$
that $\gamma^{n,k}_{[0]}=0\gamma^{n-1,k-1}_{[0]}\in\sigma_{k}\left(1\gamma^{n-1,k-1}_{[2^{n-1}-1]}\right)\subseteq\sigma_k\left(\gamma^{n,k}_{[2^{n-1}-1]}\right)$, hence $\gamma^{n,k}$ satisfies Cond. \ref{ivb}.
Similarly,   $\underline\gamma^{n-1,k-1}_{[0]}\in\sigma_{k-1}\left(\underline\gamma^{n-1,k-1}_{[2^{n-1}-1]}\right)$ implies
$\underline\gamma^{n,k}_{[0]}\in\sigma_k\left(\underline\gamma^{n,k}_{[2^{n-1}-1]}\right)$, hence $\underline\gamma^{n,k}$ satisfies Cond. \ref{ivb}.
\end{proof}\noindent 
We have now examined each of the different possibilities.
The following statement summarizes the results. 
\begin{theorem} 
\label{H-maximal}
 
Given a finite alphabet  $A$,  $k\ge 1$, and $n\ge k$, 
there is a loopless algorithm that allows to compute some specific maximum length $\sigma_k$-Gray cycle.
In addition 
exactly one the following conditions  holds:
\begin{eqnarray}
\lambda_{A,\sigma_k}(n)=\left\{
\begin{array}{ccccc}
|A^n|&|A|\ge 3, n\ge k\\
2&|A|=2, n=k\\
|A|^n&~|A|=2, n\ge k+1, k~ {\rm is~ odd}\\
~~~ |A|^{n-1}&~~~ |A|=2, n\ge k+1,  k~ {\rm is~ even.}
\end{array}
\right.
\nonumber
\end{eqnarray}
\end{theorem}
\begin{proof}
Notice that, in the case where we have $|A|=2$, with $n$ being an  even integer, according to Eqs. (\ref{Construct}), and Proposition \ref{Gamma-n-k-odd}, 
Algorithm 2 can be easily extended in  a method  computing  $\gamma^{n,k}_{[i]}$ by starting with $\gamma^{n,k}_{[i-1]}$.
As a consequence, according to the studies in Sects. \ref{A>3bis}, \ref{S4bis}, in any case there is an iterated-basis algorithms generating some specific maximum length $\sigma_k$-Gray cycle.
As indicated above, according to the implementation of $h^{n_0,1}$, $\gamma^{n_0,1}$, and $\rho^{n_0,1}$, that algorithm can run in constant amortized-time or in constant time.

In what follows, we examine the length of the corresponding $\sigma_k$-Gray cycles.
Recall that if some  $\sigma_k$-Gray cycle exists over $X\subseteq A^{\le n}$, necessarily $X$ is a uniform set that is, the inclusion $X\subseteq A^m$ holds for some $m\le n$, whence
 in any case we have $\lambda_{A,\sigma_k}(n)\le |A|^n$.

-- According to Proposition \ref{H-maximal},  if we have $|A|\ge 3$ and  $n\ge k$, a $\sigma_k$-Gray cycle exists over $A^n$, whence we have $\lambda_{A,\sigma_k}(n)=[A|^n$.

-- Similarly, according to  Proposition  \ref{Gamma-n-k-odd}, the cond. $|A|=2$, $n\ge k+1$, $k$ is odd implies that a  $\sigma_k$-Gray cycle exists over $A^n$, hence
 we have $\lambda_{A,\sigma_k}(n)=[A|^n$.

-- As indicated in the preamble of Sect. \ref{S4},  the cond. $|A|=2$ with $n=k$ trivially implies  $\lambda_{A,\sigma_k}(n)=2$.

-- Finally, according to  Lemma \ref{pary},
given a binary alphabet  $A$, if $k$ is  even  we have $\lambda_{A,\sigma_k}(n)\le 2^{n-1}$ therefore,
according to   Proposition \ref{Gamma-n-k-even} the cond. $|A|=2, n\ge k+1$, $k$ is even implies that a $\sigma_k$-Gray cycle 
 exists in $A^{n-1}$ that is  we have $\lambda_{A,\sigma_k}(n)=|A|^{n-1}$.
\end{proof}

\smallbreak\noindent
Algorithms (c) and (d) allow to construct maximum length Gray sequences  st. the Hamming distance of two consecutive terms is exactly $k$.
We close the study by examining the case where Gray cycles are defined with a weaker constraint. 

\medbreak\noindent
{\it $k$-Gray codes}\\
These sequences are commonly defined as Gray sequences 
where two consecutive  terms have distance at most $k$.
In the context of our study, $k$-Gray cycles are actually $\Sigma_k$-Gray cycles, where we set $(w,w')\in\Sigma_k$ iff.
the Hamming distance of $w$ and $w'$ is not greater than $k$.

Note that we have  $\Sigma_k=id_{A^*}\cup \sigma_1\cup\cdots\cup\sigma_k$, thus $\sigma_1\subseteq \Sigma_k$.
As a consequence, the two  notions of $\sigma_1$-Gray cycle and $k$-Gray cycle are identical. In particular, we have $|A^n|=\lambda_{A,\Sigma_k}(n)=\lambda_{A,\sigma_1}(n)$.
Furthermore, according to Theorem \ref{H-maximal}, each of Algorithms (c), (d) generates a $k$-Gray cycle of length $|A|^n$ iff. exactly one of the two following conds. holds:

~~~~~~~~~$|A|\ge 3$, $n\ge k$

~~~~~~~~~$|A|=2$, $n\ge k+1$,  $k$ is odd.

In the case where we have $|A|=2$, $n\ge k+1$,  $k$ is even or $A=2$, $n=k$, once more according to Theorem \ref{H-maximal}, our algorithms cannot compute any $k$-Gray cycle.
\section*{Further development}
The present investigations could be done 
 in the framework of other word  binary relations $\tau$, such as
other edit relations, as defined in \cite{N21}, or relations connected to the so-called prefix or factor distances \cite{CP02,N22}.
One  could also characterize maximum length Gray cycles over $X$, with $X$ describing some noticeable  families of sets such as variable-length codes.

\bibliographystyle{plain}
\bibliography{mysmallbibFULL}

\end{document}